\documentclass[journal]{IEEEtran}
\usepackage{booktabs} 
\usepackage{soul}
\usepackage{amssymb}
\usepackage{amsfonts,amsmath}
\usepackage{mathrsfs}
\usepackage[mathcal]{eucal}
\usepackage[ruled, linesnumbered]{algorithm2e}
\usepackage[noend]{algpseudocode}
\usepackage{amsthm}
\usepackage{caption}
\usepackage{multirow}
\usepackage{subcaption}
\usepackage{cite} 
\usepackage{graphicx}
\usepackage{bmpsize}
\usepackage{color}
\usepackage{caption}
\usepackage{epstopdf}

\DeclareRobustCommand{\rchi}{{\mathpalette\irchi\relax}}
\newcommand{\irchi}[2]{\raisebox{\depth}{$#1\chi$}}
\makeatletter

\newcommand{\removelatexerror}{\let\@latex@error\@gobble}
\makeatother

\newtheorem{theorem}{Theorem}
\newtheorem{definition}{Definition}
\newtheorem{insight}{Insight}
\newtheorem{remark}{Remark}

\newcommand{\ignore}[1]{}

\ifCLASSINFOpdf
\else
\fi
\begin{document}
%
\title{HashTran-DNN: A Framework for Enhancing Robustness of Deep Neural Networks against Adversarial Malware Samples}
%
%
%

\ignore{
\author{Deqiang~Li, 
		Ramesh Baral,
        Tao Li,
        Han Wang,
        Qianmu Li,
        Shouhuai Xu %
\thanks{D. Li and Q. Li are with Nanjing University of Science and Technology.}
\thanks{R. Baral is with Florida International University.}
\thanks{T. Li was with Florida International University and Nanjing University of Posts and Telecommunications.}
\thanks{H. Wang is with University of Texas at San Antonio.}
\thanks{S. Xu is with Department of Computer Science, University of Texas at San Antonio, 1 UTSA Circle, San Antonio, 78249 USA e-mail:{shxu@cs.utsa.edu} }}
}      

\author{Deqiang~Li, 
		Ramesh Baral,
        Tao Li,
        Han Wang,
        Qianmu Li, and 
        Shouhuai Xu %
\thanks{D. Li and Q. Li are with Nanjing University of Science and Technology. R. Baral is with Florida International University. T. Li is with Florida International University and Nanjing University of Posts and Telecommunications. H. Wang and S. Xu are with Department of Computer Science, University of Texas at San Antonio. Correspondence:{\tt shxu@cs.utsa.edu}}}

\maketitle

\begin{abstract}
Adversarial machine learning in the context of image processing and related applications has received a large amount of attention. However, adversarial machine learning, especially adversarial deep learning, in the context of malware detection has received much less attention despite its apparent importance. In this paper, we present a framework for enhancing the robustness of Deep Neural Networks (DNNs) against adversarial malware samples, dubbed {\em \underline{Hash}ing \underline{Tran}sformation \underline{D}eep \underline{N}eural \underline{N}etworks} (HashTran-DNN). The core idea is to use hash functions with a certain {\it locality-preserving} property to transform samples to enhance the robustness of DNNs in malware classification. The framework further uses a Denoising Auto-Encoder (DAE) regularizer to reconstruct the hash representations of samples, making the resulting DNN classifiers capable of attaining the locality information in the latent space. We experiment with two concrete instantiations of the HashTran-DNN framework to classify Android malware. Experimental results show that four known attacks can render standard DNNs useless in classifying Android malware, that known defenses can at most defend three of the four attacks, and that HashTran-DNN can effectively defend against all of the four attacks.
\end{abstract}


%
\IEEEpeerreviewmaketitle

\section{Introduction}

Malware is a major threat to cyber security, and the problem is becoming increasingly severe. For example, Symantec reports that about 355 millions, 357 millions, and 669 millions of malware variants were seen in the years of 2015, 2016, and 2017, respectively \cite{symantec:Online}. Kaspersky reports that malware attacked 2,871,965 and 1,126,701 devices in 2016 and 2017, respectively \cite{GREAT:Kaspersky,kasparsky:Online}.
This calls for effective solutions for detecting and classifying malware.

Machine learning has been widely used for malware detection and classification \cite{DBLP:journals/csur/YeLAI17}.
However, malware classifiers are susceptible to the attacks of \textit{adversarial malware examples} \cite{Biggio:Evasion, rndic_laskov,szegedyZSBEGF13,papernotMGJCS16,grossePM0M16,DBLP:conf/ijcai/HouYSA18,DBLP:conf/kdd/FanHZYA18,YeFOSINT-SI-2018,DBLP:conf/asunam/HouSCYB17,DBLP:conf/eisic/ChenYB17}. Adversarial samples can be obtained by perturbing (i.e., manipulating) a few features of malware samples that would be detected as malicious. However, these adversarial samples, while malicious, would be classified as benign.

Adversarial samples are a common threat, rather than specific to certain machine learning models or datasets
\cite{szegedyZSBEGF13, PapernotMG16, liu_2016}. 
In a broader context, the problem is known as {\em adversarial machine learning}, which is relevant to a range of domains (e.g., image processing, text analysis, and malicious website classifiers~\cite{grossePM0M16, goodfellow_2015, carliniW16a, wang_2017, PapernotMG16, xu2014evasion,DBLP:journals/corr/KurakinGB16a,takeru_2015,ens_adv_training, buckman2018}). 
Despite its clear importance, adversarial malware detection has not received the due amount of attention. This is true despite the recent studies~\cite{PapernotMG16, grossePM0M16, Hu2017, DBLP:journals/corr/RosenbergSRE17, wang_2017,DBLP:conf/ijcai/HouYSA18,DBLP:conf/kdd/FanHZYA18,YeFOSINT-SI-2018} that show how adversarial samples can easily evade malware classifiers.
The state-of-the-art is that there are no effective defenses \cite{wang_2017}. 
In this paper, we investigate a new defense against adversarial malware samples.

\smallskip

\noindent{\bf Our contributions}.
In this paper, we make three contributions.
First, we present a framework for enhancing the robustness of Deep Neural Networks (DNNs) against adversarial malware samples, dubbed {\em \underline{Hash}ing \underline{Tran}sformation \underline{D}eep \underline{N}eural \underline{N}etworks} (HashTran-DNN). The core idea is to use {\it locality-preserving} hash functions to transform samples to reduce, if not remove, the impact of adversarial perturbations.

Second, we propose using a Denoising Auto-Encoder (DAE) to regularize DNNs and reconstruct the hash representations of samples. This enables the resulting DNN classifier to capture the locality information in the latent space.
Moreover, the DAE can detect the out-of-distribution samples that are far from the support of the underlying distribution of the training data (i.e., filtering adversarial samples resulting from large perturbations).

Third, 
we introduce the notion of Locality-Nonlinear Hash (LNH) functions and presents a concrete construction that achieves a bounded distance-distortion property in the cube $\{0,1\}^n$ with respect to the {\it normalized Hamming} distance metric.
We conduct systematic experiments with a real-world dataset. Some of the findings are highlighted as follows.
\begin{itemize}
\item Standard DNNs for Android malware classification can be ruined by adversarial samples generated by the following four attacks: the {\em Jacobian-based Saliency Map Attack} (JSMA) ~\cite{papernot_2016,grossePM0M16}; the {\em Gradient Descent with Kernel Density Estimation} attack (GD-KDE)~\cite{Biggio:Evasion}; the Carlini-Wagner (CW) attack \cite{carliniW16a}; and the {\em Mimicry} attack~\cite{rndic_laskov}.

\item HashTran-DNN can substantially enhance the robustness of DNN classifiers against the four attacks mentioned above.
This robustness enhancement can be attributed to the fact that HashTran-DNN combines hash functions and DAEs to make DNN classifiers capable of attaining the locality information of samples in the latent space, rejecting the out-of-distribution samples, and defeating attacks that attempt to manipulate ``important'' features only (e.g., the CW attack) or manipulate features arbitrarily to a large extent (e.g., the Mimicry attack). 

\item With respect to the four attacks mentioned above, HashTran-DNN is more robust than the defense mechanism known as Random Feature Nullification (RFN) \cite{wang_2017}, which is not effective against any of the four attacks mentioned above.
With respect to the JSMA, GD-KDE, and CW attacks, HashTran-DNN is comparable to the iterative Adversarial Training defense \cite{DBLP:journals/corr/KurakinGB16a}, which however assumes that the defender knows the adversarial samples generated by the attackers. Moreover, HashTran-DNN is more robust than the iterative Adversarial Training defense against the Mimicry attack.
\end{itemize}

\noindent{\bf Paper outline}. 
The rest of the paper is organized as follows. Section \ref{sec:related-work} discusses the related prior studies. Section \ref{sec:background} reviews some preliminary knowledge. Section \ref{LSH} presents locality-preserving hash functions. Section \ref{sec:HashTran} describes the HashTran-DNN framework. Section \ref{sec:evl} presents our experiments and results. Section \ref{sec:lmt} discusses the limitations of the present study. Section \ref{sec:conclusion} concludes the paper.

\section{Related work}
\label{sec:related-work}
Deep learning is successful in image processing \cite{krizhevsky2012imagenet}, natural language processing~\cite{sutskever2014sequence}, and speech recognition~\cite{wavenet45774}. Since our focus is on improving the robustness of DNN-based malware classifiers, we emphasize on this topic.

From an attacker's perspective, there are three types of attacks: {\em black-box} vs. {\em white-box} vs. {\em gray-box}. In the black-box attack model, the attacker only has black-box access to the classifier; in the white-box attack model, the attacker knows everything about the defender's model; the gray-box attack model resides in between (e.g., the attacker has access to the defender's training dataset, feature set, and some information about the defender's DNN architecture). In this paper, we will focus on the gray-box model, especially the aforementioned four attacks (see Section \ref{adv_smps} for details).

From a defender's perspective, Figure~\ref{fig:01} highlights three defense approaches: {\em adjusted input}, {\em adjusted training procedure}, and {\em adjusted network architecture}, which are elaborated below. 

\begin{figure}[!htbp]
	\centering
	\scalebox{0.55}{
	\includegraphics{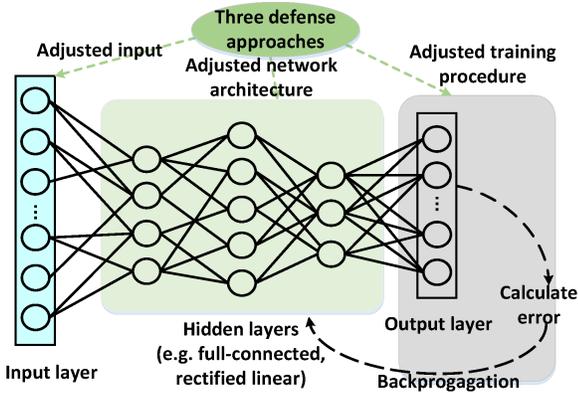}
	}
\caption{Defense approaches against adversarial samples.}
	\label{fig:01}
\end{figure}

The {\em adjusted input} approach aims to transform an input image to reduce its vulnerability to perturbation by, for example, retraining~\cite{DBLP:journals/corr/KurakinGB16a,adv_training_stability} with known adversarial samples and penalizing perceptible adversarial spaces.
The weakness of this approach is that the defender needs to know adversarial samples at the training time \cite{wang_2017,carliniW16a}, 
while noting that ensemble retraining~\cite{ens_adv_training} may alleviate the problem somewhat.
A related method 
\cite{gur_2014}
is to apply a generative model to clean up the distortions and uses a joint stacked Auto-Encoder (AE) to preprocess an input. However, the AE itself is vulnerable to adversarial samples \cite{gur_2014}. Feature squeezing~\cite{xu2017feature} and Thermometer encoding~\cite{buckman2018} can cope with image pixels with quantization strategies, but cannot deal with binary features.

This approach has been adopted to train malware detectors~\cite{Biggio6494573,xu2014evasion,wang_2017}. For example, Wang \textit{et al.}~\cite{wang_2017} introduced the idea of Random Feature Nullification (RFN), which randomly nullifies some features to make DNN classifiers non-deterministic.
However, sophistical attacks can confound with the RFN defense because it cannot nullify all of the ``important'' features, which may be exploited by the attacker.
HashTran-DNN also utilizes randomness to thwart adversarial malware samples, but is more robust than the RFN defense.

The {\em adjusted training procedure} approach aims to identify the optimal resistance against adversarial samples.
Goodfellow \textit{et al.}~\cite{goodfellow_2015} describe adversarial training as a regularization term for decreasing the generalization error.
This idea is later extended to the setting of semi-supervised learning \cite{takeru_2015}.
A limitation of this approach is also that the defender does not know all adversarial samples. 
Another idea \cite{mengc_2017,DBLP:journals/corr/GrosseMP0M17} is to treat adversarial examples as an extra category,
which is however ineffective \cite{DBLP:journals/corr/abs-1711-08478}. Inspired by the observation that large singular values in the weight matrices contribute to the vulnerability of DNNs~\cite{szegedyZSBEGF13}, yet another idea~\cite{DBLP:conf/icml/CisseBGDU17} is to use \textit{parseval networks}. 
In contrast to these studies, HashTrah-DNN uses DAEs 
to tune parameters at the hidden layers to decrease DNNs' sensitivity to adversarial perturbations.

The {\em adjusted network architecture} approach aims to adjust the architecture of the hidden layers to defend against adversarial samples. 
Krotov \textit{et al.}~\cite{DBLP:journals/corr/KrotovH17} propose the idea of \textit{Dense Associative Memory}, which uses higher-than-quadratic-order activation functions.
Another method 
\cite{DBLP:journals/corr/PapernotMWJS15} is to use a distillation mechanism to compress the vanilla model into a small network, but is known to be vulnerable \cite{carliniW16a}.

\section{Preliminaries} \label{sec:background}

In order to improve readability, Table \ref{table:notations} summarizes the main notations that are used throughout the paper.
\begin{table}[htbp!]
\centering\caption{Summary of notations\label{table:notations}}
\begin{tabular}{l|p{.32\textwidth}}
\hline
Notation & Meaning\\\hline
$n$& the number of dimensions of data samples\\
$o$& the number of classes \\
$d$& height of decision binary trees\\
$d_H(\cdot,\cdot)$ & Hamming distance  \\
$\bar{d}_H(\cdot,\cdot)$ & normalized Hamming distance  \\
$\vec{x},\vec{x_1},\vec{x_2} \in \mathbb{R}^n$  & samples represented as vectors\\
$x_i \in \mathbb{R}$ & the {\it i}-th component of $\vec{x}$ \\
$\delta_{\vec{x}} \in \mathbb{R}^n$ & adversarial perturbation to $\vec{x}$\\
$\vec{x}' \in \mathbb{R}^n$ & adversarial sample and $\vec{x}'=\vec{x}+\delta_{\vec{x}}$\\
$\epsilon$& upper bound of perturbations, i.e., $\|\delta_{\vec{x}}\|_0\leq \epsilon$\\
$y_{i} \in [o]$ & ground truth label of $\vec{x}_i$, $[o]=\{1,2,\dots,o\}$  \\
$\vec{y}_{i} \in \{0,1\}^o$ & one-hot encoding ground truth label of $\vec{x}_i$\\
$\rchi \subseteq \mathbb{R}^n \mathop{\times} [o]$& training data set $\{(\vec{x}_i,y_{i})\}_{i=1}^N$\\
$Z: \mathbb{R}^n \to {\mathbb{R}}^o$& a DNN (including its softmax layer)\\
$\vec{\hat{y}} \in {\mathbb{R}}^o$& the output of a DNN on input sample $\vec{x}$\\
$C: {\mathbb{R}}^o\to [o]$& classification labels based on $\vec{\hat{y}}$\\
$\mathcal{L}(\theta; \vec{x}, \vec{y})$  & cross entropy with parameters $\theta$, feature vector $\vec{x}$ and label vector $\vec{y}$ \\
$\{H_i\}_{i\in I}$&a family of hash functions \\
$\leftarrow_R$ & sampling from a set uniformly at random (with replacement) \\
$h_i: \mathbb{R}^n \to \mathbb{R}$ & a hash function sampled from $\{H_i\}_{i\in I}$\\
$g_{\text{LSH}}^K$ & a vector of $K$ locality-sensitive hash functions (LSH), i.e., $g_{\text{LSH}}^K=[h_1, h_2, \cdots, h_K]$\\
$\text{DT}^{m,d}_{i,j}$ & decision tree function with height $d$ . $\text{DT}^{m,d}_{i,j}: \mathbb{R}^m \to \{0,1\}^{2^{d-1}}$, where $i$ and $j$ are indices \\
$g_{\text{LNH}}^K$ & a vector of locality-nonlinear hash (LNH) functions, i.e., $g_{\text{LNH}}^K=[\text{DT}^{m,d}_{i,1},\dots,\text{DT}^{m,d}_{i,K}]$  \\
$\mathcal{H}$ & a family of hashing transformations, e.g., $g^K_{\text{LSH}}$ and $g^K_{\text{LNH}}$\\
${\bf H}_{\text{LSH}}$ or LSH &  a family of $g_{\text{LSH}}^K$ hashing transformations\\
${\bf H}_{\text{LNH}}$ or LNH & a family of $g_{\text{LNH}}^K$ hashing transformations \\
${\bf M}_{\mathcal{H}}(\vec{x})$ & Matrix representation for sample $\vec{x}$ under $\mathcal{H}$ \\ 

\hline
\end{tabular}
\end{table}

\subsection{Deep feed-forward neural networks} 
\label{DNNs}

In this paper, we focus on DNNs with a \textit{softmax} layer and $l$ hidden layers. A DNN classifier takes an input $\vec{x}\in\mathbb{R}^n$ and produces an output $\vec{\hat{y}} = Z(\vec{x})=[Z_1(\vec{x}),\dots,Z_o(\vec{x})]\in\mathbb{R}^o$, where $o$ is the number of classes (e.g., $o=2$ for malware classification). The output 
\begin{align}
&Z(\vec{x})=\textrm{softmax}(F(\vec{x})) \label{Eq:dnn}\\
~~~~&\textrm{where}~F(\vec{x})= F_l(\cdots F_2(F_1(\vec{x};\theta_1);\theta_2)\cdots) \label{ep:nn_nosoftmax}
\end{align}
gives the probabilities that sample $\vec{x}$ respectively belongs to one of the $o$ classes, where
\begin{equation}
F_i(\vec{x},\theta_i)=\sigma(\theta_i\cdot\vec{x}+b_i)
\end{equation}
with some non-linear {\em activation function} $\sigma$ (e.g., sigmoid, ReLU~\cite{schmidhuber2015deep}, or ELU~\cite{clevert2015fast}), {\em weight matrix} $\theta_i$, and {\em bias} $b_i$.
Let 
\begin{equation}
\label{definition:classifying}
C(Z(\vec{x})) = \arg\max_i(Z_i(\vec{x})), ~1\leq i \leq o
\end{equation}
denote the class (or label) a DNN classifier assigns  to sample $\vec{x}$, where $\arg\max\limits_i(\cdot)$ returns the index of the class that has the maximum probability.  At the training phase, a loss function 
is minimized 
via backpropagation (see, for example, \cite{lecun2015deep,kingmaB14}). Moreover, we consider a modified DNN $Z'$ which takes a matrix input $M=[\vec{m}_1;\dots;\vec{m}_L]$ with each $\vec{m}_i$ as a row vector, and
\begin{align}
&Z'(M)= \nonumber \\
&\textrm{softmax}(F_l(\cdots F_2(F_1^1(\vec{m}_1;\theta_1^1),\dots,F_1^L(\vec{m}_L;\theta_1^L);\theta_2)\cdots)). \label{Eq:dnn2}
\end{align}

\subsection{Attacks for generating adversarial samples} \label{adv_smps}

An adversarial sample is represented as $\vec{x}'=\vec{x}+\delta_{\vec{x}}$, where $\vec{x}$ is the original sample and $\delta_{\vec{x}}$ is a perturbation vector. Then, 
\begin{equation}
 {\sigma}({\theta_i}{\cdot}{\vec{x}'}+b_i)={\sigma}({\theta_i}{\cdot}{\vec{x}} + {\theta_i}{\cdot}{\delta_{\vec{x}}}+b_i) \textrm{,}  \label{eq:linear}
\end{equation}
where ${\theta_i}{\cdot}{\delta_{\vec{x}}}$ is the distortion item.
In the context of malware classification, any perturbation should preserve the malicious functionality of the original sample (i.e., an adversarial sample can run in the same environment to cause damages).

\smallskip

\noindent{\bf Small vs. large perturbations}.
We distinguish adversarial samples based on the degree of perturbation, {\em small} vs. {\em large}, because they can be treated differently.
On one hand, the fact that some elements of the weight matrix $\theta_i$ are overly large~\cite{goodfellow_2015}, and therefore can be exploited by the attacker to craft a slight perturbation vector $\delta_{\vec{x}}$ such that $C(Z(\vec{x}'))\neq C(Z(\vec{x}))$, where ``slight perturbation'' means that the degree of perturbation is bounded by a certain norm (e.g., $\ell_0$, $\ell_2$, or $\ell_\infty$ norm) such that $\vec{x'}$ is not far from $\vec{x}$. 
On the other hand, an attacker can arbitrarily manipulate the original sample $\vec{x}$ to generate an adversarial version $\vec{x'}$ that is far from the underlying data distribution. 

\smallskip

\noindent{\bf Four attacks in the gray-box model}.
As mentioned above, we focus on the \textit{gray-box} attack model, 
in which the attacker can train DNN classifiers on its own and then leverage them to generate adversarial samples~\cite{szegedyZSBEGF13,goodfellow_2015, carliniW16a,DBLP:ofs_abs-1802-00420}. The {\it transferability} property of machine learning contributes to the effectiveness of these attacks~\cite{szegedyZSBEGF13, papernotMGJCS16, PapernotMG16, papernot_2016, liu_2016, mengc_2017}. 

\subsubsection{Jacobian-based Saliency Map Attack (JSMA)} 
This is a {\em gradient-based} attack~\cite{papernot_2016}, in which
the attacker looks for the optimal perturbation based on the Jacobian matrix of the DNN feed-forward function with respect to an input $\vec{x}$, namely
\begin{equation}
J_Z(\vec{x})=\frac{\partial{Z(\vec{x})}}{\partial{\vec{x}}}= [{\frac{\partial{Z_j(\vec{x})}}{\partial{x_i}}}]_{i \in {1\cdots n}, j \in {1 \cdots o}}. \label{jm}
\end{equation}
In order to make the target DNN misclassify the perturbed version of $\vec{x}$, the attacker can leverage the saliency map
\begin{equation}
S(\vec{x}, y')[i]=
\begin{cases}
0 ~~~~~ \mbox{\text{if} $\frac{\partial{Z_{y'}(\vec{x})}}{\partial{x_i}} < 0$ or $\sum_{j\neq y'}\frac{\partial{Z_j(\vec{x})}}{\partial{x_i}} > 0$} \nonumber\\

(\frac{\partial{Z_{y'}(\vec{x})}}{\partial{x_i}})|\sum_{j\neq y'}\frac{\partial{Z_j(\vec{x})}}{\partial{x_i}}|\quad\mbox{otherwise},\nonumber
\end{cases}
\end{equation}
such that $\vec{x}$ is perturbed at locations $i$ if $S(\vec{x}, y')[i]$ gives the largest value. The perturbation maximizes the changes of the DNN classifier outputs in the desired output direction. This attack has been used against DNN-based malware classifiers~\cite{grossePM0M16}.

\subsubsection{Gradient Descent with Kernel Density Estimation (GD-KDE) attack}
This is an {\em optimization-based} attack \cite{Biggio:Evasion}, in which the attacker attempts to find the optimal adversarial sample $\vec{x}'$ that minimizes the following objective: 
\[\begin{aligned}
\min \limits_{x'}~\hat{g}(\vec{x'}) - \frac{\lambda}{N_t} \sum \limits_{i|y_i=y'}^{N_t} k(\vec{x'},\vec{x_i}), ~~\text{subject to} ~~ \|\vec{x'} - \vec{x}\| < \epsilon,
\end{aligned}
\]
where $\hat{g}(\vec{x'})$ estimates the cost of the posterior probability of the target label $y'$ with $y' \neq y$, $k(\cdot, \cdot)$ is a kernel density estimator (e.g., Laplacian kernel) for lifting $\vec{x'}$ to the populated region of target samples, 
$\lambda$ is the weight factor, $N_t$ is the number of target samples, and $\|\cdot\|$ refers to a norm of interest. 

\subsubsection{Carlini-Wagner (CW) attack}
This is an {\em optimization-based} attack~\cite{carliniW16a}, in which the attacker attempts to find an adversarial sample $\vec{x}'$ such that the perturbation vector $\delta_{\vec{x}}$ is minimized and the classifier misclassifies $\vec{x}'$, leading to the following formulation:
\begin{align}
	&\min \limits_{\delta_{\vec{x}}}{\lVert{\delta_{\vec{x}}}\rVert}_2^2 + {\lambda}f(y,\vec{x}+\delta_{\vec{x}})\textrm{,} ~~\textrm{where} \label{eq:cw} \\
	f(y,\vec{x}+\delta_{\vec{x}})&=\max\left\{F(\vec{x}+\delta_{\vec{x}})_y-\max_{i\neq y}\left\{F(\vec{x}+\delta_{\vec{x}})_i\right\},-\iota\right\},\nonumber
\end{align}
where $F(\cdot)$ is the output of a DNN prior to the softmax layer, $\iota$ is a scalar controlling the mis-classification confidence, $\lambda$ is the penalization factor. Since the $\ell_0$-norm is not differentiable, the $\ell_2$-norm can be used instead.

\subsubsection{Mimicry attack}
In this attack~\cite{rndic_laskov}, the attacker attempts to modify a malware sample $\vec{x}$ into an adversarial sample $\vec{x'}$ such that $\vec{x'}$ mimics a chosen benign sample as much as possible. This attack is applicable to any classifiers because it does not require the attacker to know the defender's machine learning algorithm.

\subsection{Two defenses proposed in the literature}
\label{sec:review-defense-methods}

We will compare HashTran-DNN with two defense methods.
The first defense method is called {\em Random Feature Nullifications} (RFN) \cite{wang_2017}, which randomly nullifies features at both the training phase and the testing phase. Specifically, given (i) a batch of $N$ training samples $\{\vec{x_i}\}_{i=1}^{N}$ and their one-hot encoding labels $\{\vec{y_i}\}_{i=1}^{N}$, and (ii) a random feature nullification function $f_F$, the defense aims to minimize the following objective function:
$$
\min \limits_{\theta}\frac{1}{N}\sum\limits_{i=0}^{N}\mathcal{L}(\theta\mathrm{;}\,Z(f_F(\mathbf{0}_i\mathrm{,}\vec{x_i}))\mathrm{,} \vec{y_i}), \label{obj_fun_RFN}
$$
where $\mathbf{0}_i=\lceil{n\times{p_i}}\rceil$ is the number of nullified features in input $\vec{x_i}$, the probability $p_i$ is sampled from a Gaussian distribution, and $\lceil\cdot\rceil$ is the ceiling function.

The second defense method is called {\em Adversarial Training}~\cite{goodfellow_2015,szegedyZSBEGF13,DBLP:journals/corr/KurakinGB16a},
and can be used for most machine learning algorithms. The {\em Iterative Adversarial Training} method ~\cite{DBLP:journals/corr/KurakinGB16a} aims to minimize the following cost function: 
$$
\resizebox{0.48\textwidth}{!} {$\min \limits_{\theta}\frac{1}{N_1+\lambda(N-N_1)}\left(\sum\limits_{i=0}^{N_1}\mathcal{L}(\theta\mathrm{;}~Z(\vec{x_i})\mathrm{,}\vec{y_i}) +  \lambda\sum\limits_{i=N_1}^{N}\mathcal{L}(\theta\mathrm{;}~Z(\vec{x_i}')\mathrm{,}\vec{y_i})\right)$}, \label{obj_fun_adv_train}
$$
where $\vec{x_i}'$ is an adversarial example perturbed from $\vec{x_i}$, $N_1$ is the number of unperturbed samples in the training set, $\lambda$ strengths the penalization for adversarial mis-classifications. A similar idea, called {\em proactive training}, was investigated in \cite{xu2014evasion} with respect to decision-tree classifiers .

\section{Locality-Preserving Hash Functions}
\label{LSH}

In this section we first review Locality-sensitive hashing (LSH) and then introduce locality-nonlinear hashing (LNH).

\subsection{LSH} 

LSH \cite{gionis_1999} is a family of hash functions, denoted by $\{H_i\}_{i\in I}$ where $H_i:\mathbb{R}^n\to \mathbb{R}$ with the following property: 
For a fixed $H_i$ (determined by index $i$), two ``nearby'' inputs are mapped to the same hash value with a high probability, but two ``distant'' inputs are mapped to the same hash value 
with a small probability, where the distance can be \textit{Jaccard}, \textit{Hamming} (based on the $\ell_0$-norm, and denoted by $d_H$), or the $\ell_p$-norm. 
In the present paper, we focus on the $\ell_0$-norm because the datasets use a binary representation of malware features, leading to the {\it Hamming} space. 
Formally, we have:

\begin{definition}[LSH hash functions \cite{gionis_1999}]
A LSH function $H_i(\cdot)$ has the following {\em locality-sensitivity} property:
For two inputs  $\vec{x}_1$ and $\vec{x}_2$ such that $d_H(\vec{x}_1,\vec{x}_2) \leq \epsilon$ for some $\epsilon$, 
the probability $\Pr(H_i(\vec{x}_1)=H_i(\vec{x}_2))$ is large;
otherwise, $\Pr(H_i(\vec{x}_1)=H_i(\vec{x}_2))$ is small.
\end{definition}
 
An example of LSH is the following {\cite{gionis_1999}}. Consider the $Hamming$ distance over a bit vector $\vec{x}\in\{0,1\}^n$, LSH functions $\{H_i\}_{i\in I}$ can be constructed from the {\em bit sampling} method \cite{gionis_1999}, which randomly selects a bit from the input $\vec{x}$ as the hash value.
However, this construction has two weaknesses: 
(i) It is a linear transformation, and therefore vulnerable to adversarial examples~\cite{goodfellow_2015}.
(ii) It leads to linearly correlated hash values when applied to samples that are overly sparse in $\{0,1\}^n$, which can undermine the {locality-sensitivity property} and therefore its usefulness in defending against adversarial samples.

In order to enhance the locality-sensitivity property, we can use a vector of $K$ LSH functions,
denoted by 
$$g^K_{\text{LSH}}=[h_1\textrm{,}~h_2\textrm{,}~\ldots\textrm{,}~h_K],$$ 
where $h_j \leftarrow_R \{H_i\}_{i\in I}$ for $1\leq j \leq K$ and ``$\leftarrow_R$'' means sampling uniformly at random (with replacement).
This leads to a hashing transformation 
$$g^K_{\text{LSH}}(\vec{x})=[h_1(\vec{x})\textrm{,}~h_2(\vec{x})\textrm{,}~\ldots\textrm{,}~h_K(\vec{x})].$$
We can repeat the aforementioned sampling process,
leading to $L$ independent $g^K_{\text{LSH}}$ functions, denoted by 
$${\mathbf H}_{\text{LSH}}=\{g^K_{\text{LSH},1};g^K_{\text{LSH},2};\ldots;g^K_{\text{LSH},L}\}.$$
When the meaning is clear from the context, we may use LSH and ${\bf H}_{\text{LSH}}$ interchangeably to simplify the presentation.

\begin{algorithm}
\KwIn{Training data $\rchi=\{(\vec{x},y)\}$, where $y$ is the ground truth label of $\vec{x}$; LSH family $\{H_i\}_{i\in I}$; $d$ (Decision Tree height); $m$ (the length of random feature sub-vectors for training a Decision Tree); $L$ (number of hashing transformations); $K$ (number of Decision Trees used in a hashing transformation)}
\KwOut{${\bf H}_{\text{LNH}}(\vec{x})$, which is a binary matrix of $L$ rows and $K\times 2^{d-1}$ columns}

\For{$i=$ \rm{1} to ${L}$}{
\For{$j=$ \rm{1} to ${K}$}{
Choose $m$ LSH functions $h_1\mathrm{,}h_2\mathrm{,}\ldots\mathrm{,}h_m \leftarrow_R \{H_i\}_i$;

\For{$(\vec{x},y) \in \rchi$}{
define $[h_1(\vec{x})\mathrm{,}~h_2(\vec{x})\mathrm{,}~\ldots\mathrm{,}~h_m(\vec{x})]$ as feature representation of $\vec{x}$;
}

Train a full-binary Decision Tree $\text{DT}_{i,j}^{m,d}$ of height $d$ (and $2^{d-1}$ leaves) from the transformed data $\{[h_1(\vec{x})\mathrm{,}h_2(\vec{x})\mathrm{,}\ldots\mathrm{,}h_m(\vec{x})],~y\}_{(\vec{x},y) \in \rchi}$;

Label the leave of Decision Tree $\text{DT}_{i,j}^{m,d}$ corresponding to the path $\text{DT}_{i,j}^{m,d}(\vec{x})$ as ``1'' and each of the other $2^{d-1}-1$ leaves as ``0'';

}
}
\Return ${\bf H}_{\text{LNH}}(\vec{x})=[g^K_{\text{LNH},1}(\vec{x})\textrm{;}~g^K_{\text{LNH},2}(\vec{x})\textrm{;}~\ldots\textrm{;}~g^K_{\text{LNH},L}(\vec{x})]$ ~~/* a matrix of $L$ rows and $K\times 2^{d-1}$ columns */

\caption{Constructing ${\bf H}_{\text{LNH}}$ from LSH family $\{H_i\}_i$}
\label{LFH_alg}
\end{algorithm}

\subsection{LNH (Locality-Nonlinear Hashing)} 

We introduce LNH, which does not have the afore-mentioned weaknesses of {\it bit sampling}.
As shown by Algorithm~\ref{LFH_alg}, the idea is to construct a family of hashing transformations ${\bf H}_{\text{LNH}}$ from LSH functions $\{H_i\}_{i\in I}$, as follows:
\begin{itemize}
\item[(i)] Use LSH functions to transform samples $\{\vec{x}\}$ to their hashed values $[h_1(\vec{x}),h_2(\vec{x}),\ldots,h_m(\vec{x})]$ for $K$ independent times. 
\item[(ii)] Use these hashed values (i.e., has representations) of the training samples and their labels, namely $\{[h_1(\vec{x})\mathrm{,}~h_2(\vec{x})\mathrm{,}~\ldots\mathrm{,}~h_m(\vec{x})];y\}_{(\vec{x},y) \in\,\rchi}$, to train a Decision Tree $\text{DT}_{i,j}^{m,d}$ of height $d$ and $2^{d-1}$ leaves, where $1\leq j \leq K$. 
\item[(iii)] For each Decision Tree $\text{DT}_{i,j}^{m,d}$, label its leaves as follows: The leave on the path corresponding to $\text{DT}_{i,j}^{m,d}(\vec{x})$ is labeled as ``1'', and each of the other $2^{d-1}-1$ leaves is labeled as ``0''. Then, define $g_{\text{LNH}}^K=[\text{DT}^{m,d}_{i,1},\dots,\text{DT}^{m,d}_{i,K}]$ and hence, $g^K_{\text{LNH},i}(\vec{x})=[\text{leaves of DT}_{i,1}^{m,d}~\text{from left to right},\ldots$, leaves of $\text{DT}_{i,K}^{m,d}~\text{from left to right} ]$, which is a binary vector of $K\times 2^{d-1}$ elements.

\item[(iv)] Repeat (i)-(iii) for $L$ times, leading to  a family of $g_{\text{LNH}}^K$ hashing transformations, i.e., ${\bf H}_{\text{LNH}}=\{g^K_{\text{LNH},1}\textrm{;}~g^K_{\text{LNH},2}\textrm{;}~\ldots\textrm{;}~g^K_{\text{LNH},L}\}$. 
\end{itemize}
When the meaning is clear from the context, we may use use LNH and ${\bf H}_{\text{LNH}}$ interchangeably to simplify the presentation.

\begin{figure}[!htbp]
	\centering
	\scalebox{0.25}{
	\includegraphics{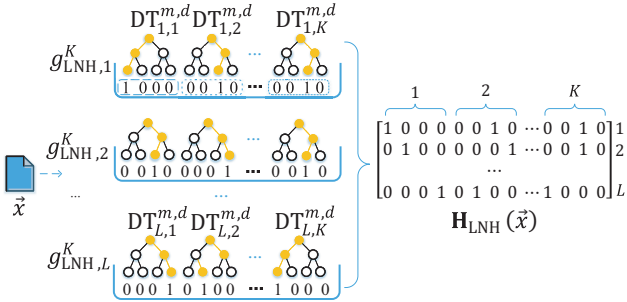}
	}
\caption{Illustration of computing ${\bf H}_{\text{LNH}}(\vec{x})$, where each Decision Tree is a full binary tree of height $d=3$.}
	\captionsetup{font={scriptsize}}
	\label{fig:LFH}
\end{figure}

Figure~{\ref{fig:LFH}} illustrates the construction of ${\bf H}_{\text{LNH}}$ and computation of ${\bf H}_{\text{LNH}}(\vec{x})$ as a binary matrix of $L$ rows and $K\times 2^{d-1}$ columns. Now we make some observations. First, the use of Decision Trees makes ${\bf H}_{\text{LNH}}(\cdot)$ nonlinear and non-differentiable. Second, when LSH functions are constructed from the {\em{bit sampling}} method, ${\bf H}_{\text{LNH}}$ can be seen as a particular kind of {\em random subspace} method~\cite{709601_rss}, which decreases the generalization error of learning-based models. Third, it is known \cite{PapernotMG16} that individual Decision Trees are vulnerable to ``cross-model'' adversarial examples crafted from other learning techniques (e.g., support vector machine, DNNs). This is no concern because we use a forest of Decision Trees, each of which {is learned from some random feature subspace}.

\begin{figure*}[!htbp]
\centering
\scalebox{0.439}{\includegraphics{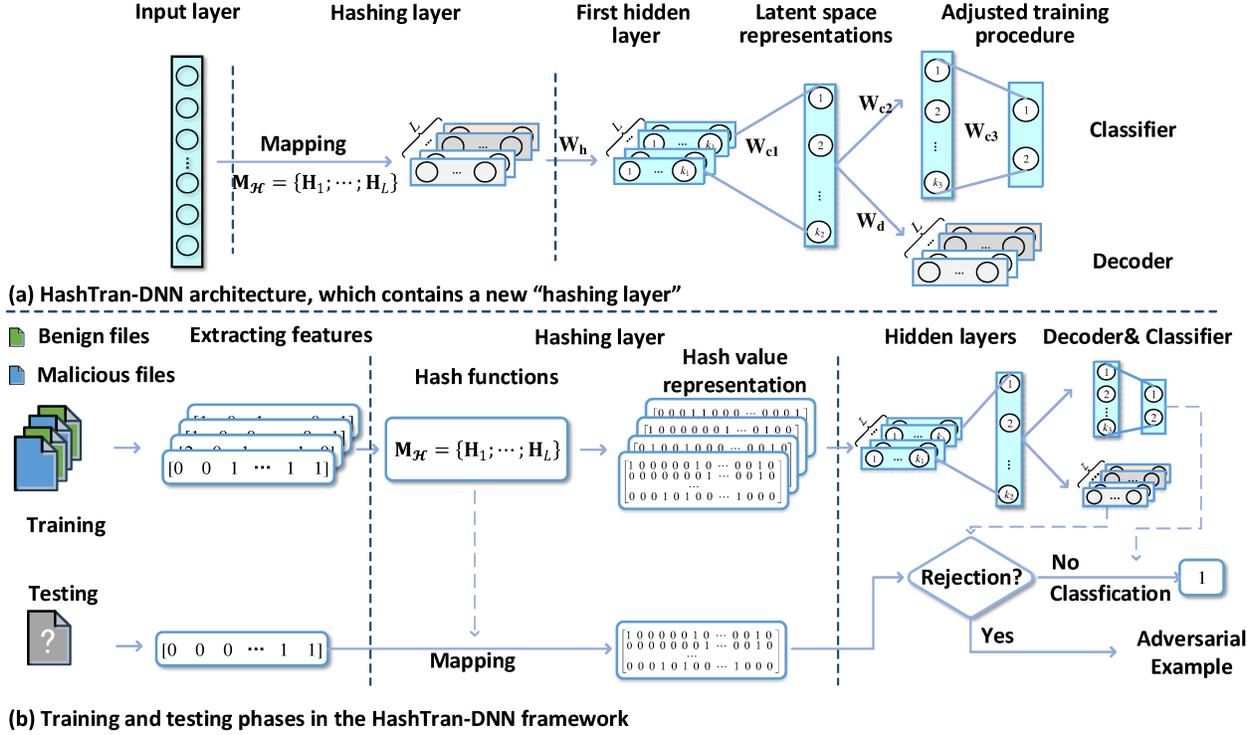}}
\caption{The HashTran-DNN framework. (a) HashTran-DNN architecture, which adds a hashing layer to a feed-forward DNN.
(b) Training and testing phases in HashTran-DNN: the training phase aims to minimize the classification error and leverage the DAE (to reconstruct the hash representations) to regularize the DNN.
The testing phase aims to reject the out-of-distribution examples based on the DAE reconstruction error and predict the class for the remaining examples by classifier.
}
\captionsetup{font={scriptsize}}
\label{fig:002}
\end{figure*}

\section{The HashTran-DNN framework} \label{sec:HashTran}

In this section, we present the HashTran-DNN framework and a theoretic analysis of it. 

\subsection{Basic idea} 
\label{basic_ideal}

The HashTran-DNN framework is centered at the idea of constructing a family of hashing transformations $\mathcal{H}=\{{\bf H}_j\}_{j\in T_{\mathcal{H}}}$ with a {\em locality-preserving property} specified by Ineq.~\eqref{geq:2Pr} below. Two examples of $\mathcal{H}$ are the aforementioned ${\bf H}_{\text{LSH}}$ and ${\bf H}_{\text{LNH}}$, meaning that ${\bf H}_j$ can be instantiated as either $g_{\text{LSH},\cdot}^K$ or $g_{\text{LNH},\cdot}^K$ and that ${\bf H}_j(\vec{x})$ returns a vector. Let us first consider the case the attacker makes small perturbations while preserving the malicious functionality of the adversarial samples.
Specifically, consider a sample $\vec{x_i}$, its label $y_i$, and an adversarial sample $\vec{x_i}'$ derived from $\vec{x_i}$ and perturbation $\|\delta_{\vec{x_i}}\|_0\leq \epsilon$ for some $\epsilon>0$, 
the defense aims to assure:
\begin{align}
\mathbf{E}[\#\{j\in T_{\mathcal{H}}:{\bf H}_j(\vec{x}_i)={\bf H}_j(\vec{x_i}')\}] \geq  \Theta \label{geq:2Pr} \\ 
C(Z'({\bf M}_{\mathcal{H}}(\vec{x}_i))))=C(Z(\vec{x}_i)),\label{eq:keeping}
\end{align}
where $\mathbf{E}[\cdot]$ is the expectation function, $\#$ denotes the cardinality of a set, $\Theta$ is the desired robustness,
${\bf M}_{\mathcal{H}}(\vec{x})$ is the matrix representation of $\mathcal{H}(\vec{x})$, $Z'$ is a newly constructed DNN that takes hashing matrix as the input (see Eq.~(\ref{Eq:dnn2})), and $C(\cdot)$ returns a DNN's prediction on labels of samples.

On one hand, Ineq.~(\ref{geq:2Pr}) indicates the {\em robustness} against adversarial samples.
Specifically, a large $\Theta$ indicates a high robustness  
because a large $\Theta$ means there are more hashing transformations $\{\mathbf{H}\}_{j\in T_{\mathcal{H}}}$ that can ``eliminate'' the effect of adversarial perturbations (i.e., the perturbations are useless to the attacker). As we will see, this property can be rigorously proved as Theorem \ref{Theo:robust} in Section \ref{sec:analysis}.

On the other hand, Eq.~(\ref{eq:keeping}) assures the {\em classification accuracy} by requiring that the newly constructed DNN assigns the same label to the hash-transformed representation ${\bf H}_j(\vec{x})$ of $\vec{x}$ as to $\vec{x}$. However, Eq.~(\ref{eq:keeping}), while intuitive, is difficult to prove. Therefore, we consider the following alternative with a weaker guarantee: Given binary feature vectors, there exist hashing transformations, ${\bf H}_j\in \mathcal{H}$, that are close to the distance-preserving transformation with respect to the normalized Hamming distance. This means that the hashing transformation would not cause much metric distortion in the Hamming cube of the original feature space. This is important because classification accuracy is highly dependent upon the underlying low dimensional structure of the samples, as reported by a recent study~\cite{jiang2018trust}. The alternate guarantee is proven as Theorem \ref{Theo:accuracy} in Section \ref{sec:analysis}.

Now we consider the case of large perturbations while preserving the malicious functionality of the adversarial samples. In this case, Ineq.\eqref{geq:2Pr} can be thwarted. 
In order to defend against such attackers, we leverage auto-encoders to detect adversarial samples that are far from the training sample space~\cite{mengc_2017}. At this stage, we are only able to empirically show the effectiveness of HashTran-DNN against large perturbations; theoretic treatment is left as an open problem.

\subsection{The framework}

Figure~\ref{fig:002} highlights the HashTran-DNN framework, 
which adds a ``hashing layer'' to a feed-forward DNN.
The training phase of HashTran-DNN has three steps: 
(i) extracting features for representing training samples; (ii) using hash functions to transform the feature representation of the training samples to $L$ vector representations (i.e., matrix representations); and
(iii) learning a DNN from the matrix representation.
Now we elaborate these steps and discuss the testing phase.
 
\subsubsection{Extracting features}
There have been numerous studies on defining features for malware detections (e.g., \cite{grossePM0M16, wang_2017, dang2017evading, Chen:2017:SES,DBLP:journals/virology/SalehLX18,DBLP:conf/milcom/SalehRX17,Sayfullina7345283}). 
These features may be extracted via static analysis, dynamic analysis, or a hybrid of them. 
Particularly for binary feature vectors, ``1'' means a feature is present in the sample and ``0'' means the feature is absent.
Many malware detectors (e.g., \cite{grossePM0M16, wang_2017,Chen:2017:SES}) use binary representations while achieving a satisfying accuracy. 

\subsubsection{Hashing layer}
The hashing layer uses some ${\mathcal H}$, such as ${\bf H}_{\text{LSH}}$ or ${\bf H}_{\text{LNH}}$, to transform the binary feature representation to the vector (or more precisely, matrix) representation. 
(In our experiments that will be presented in Section \ref{sec:evl}, each sample will be transformed to $L$ vectors $\{{\bf H}_j(\vec{x})\}_{j=1}^L$, which formulates a binary matrix ${\bf M}_{\mathcal{H}}(\vec{x})\in \mathbb{Z}_2^{L\times T}$  for some $T$.)
Assuming $k_1$ neurons in the first hidden layer handling each row vector, we have 
$L \times k_1$ neurons in the first hidden layer. Note that the row order in the matrix representation of samples does not matter, because each row is treated independently in the first layer of DNN $Z'$ before mixing them at later stages. 

\subsubsection{Learning DNNs}
At the training phase, we use the hashed vector representation to learn a DNN classifier. What is unique to the HashTran-DNN learning is that the training phase not only aims to minimize the classification error, but also leverages the DAE to regularize the DNN. Basically, the DAE encodes the training samples compactly, and reconstructs the input from the compact encoding (also known as the {\em latent space representations}). This allows the DAE to retain the locality information by reconstructing the hash representation matrix $\mathbf{M}_{\mathcal{H}}$, which is related to manifold learning~\cite{bengio2013representation} or representation learning~\cite{Goodfellow-et-al-2016}.

Specifically, the DAE maps a perturbed matrix ${\bf M}_{\mathcal{H}}(\vec{x})\oplus \Delta$
for some random noise matrix $\Delta$ 
to ${\bf M}_{\mathcal{H}}(\vec{x})$. 
In case of $\ell_0$-norm on binary features, a random subset of elements in $\Delta$ have value 1 and the other elements have value 0, where
the number of 1's is $\lceil{L\times T\times p}\rceil$ for some $p$ sampled from the Gaussian Distribution $\mathcal{N}(0,{(\epsilon/n)}^2)$. The ``corrupt'' elements of ${\bf M}_{\mathcal{H}}(\vec{x}) \oplus \Delta$ are meant to simulate adversarial perturbations. 
The learned DNN is robust to small perturbations by performing the DAE regularization upon the hash transformation, which captures the locality preserving property in the data accurately when the reconstruction error is small.

The weight set of the DAE is $\theta_d=\{\mathbf{W_{h}};\mathbf{W_{c1}}$;$\mathbf{W_d}\}$, and the weight set of the classifier is $\theta_c=\{\mathbf{W_{h}};\mathbf{W_{c1}}$;$\mathbf{W_{c2}}$;$\mathbf{W_{c3}}\}$, where $\mathbf{W_{h}}$, $\mathbf{W_{c1}}$, $\mathbf{W_{c2}}$, and $\mathbf{W_{c3}}$ are weight matrices (cf. Figure~\ref{fig:002}).
During the training phase, we treat DAE as part of the HashTran-DNN framework and train all the weights jointly. Given a training set $\{(\vec{x_i},y_i)\}_{i=1}^N$ and a hashing transformation $\mathcal{H}$, we convert the scalar label $\{y_i\}_{i=1}^N$ into the one-hot encoding labels $\{\vec{y}_i\}_{i=1}^N$ and consider the classification loss function
\begin{equation}
\mathcal{L}_{\scriptsize{C}} = \frac{1}{N}\sum\limits_{i=0}^{N}\mathcal{L}(\theta_c\mathrm{;}\,(Z'(\vec{x_i}). \vec{y_i})). \label{obj_fun}
\end{equation}
The widely-used DAE reconstruction loss function $\mathcal{L}_{\scriptsize{D}}$ is 
\begin{equation}
\label{eq:other-LD}
\mathcal{L}_{\scriptsize{D}} = \frac{1}{N}\sum\limits_{i=0}^{N}\parallel{\text DAE}({\bf M}_{\mathcal{H}}(\vec{x_i}) \oplus \Delta_i) - {\bf M}_{\mathcal{H}}(\vec{x_i})\parallel_2^2,
\end{equation}
where $\text{DAE}(\cdot)$ is the output of DAE whose activation function at the last layer is the sigmoid. However, the definition given by Eq.\eqref{eq:other-LD} is not suitable for the setting of the present paper because the hash representation is binary. As such, we use the cross-entropy function
\begin{align}
\mathcal{L}_D &=-\frac{1}{N}\sum\limits_{i=0}^{N} [{\bf M}_{\mathcal{H}}(\vec{x_i})\log({\text DAE}({\bf M}_{\mathcal{H}}(\vec{x_i}) \oplus \Delta)) \nonumber \\
& + (1 - {\bf M}_{\mathcal{H}}(\vec{x_i}))\log(1 - {\text DAE}({\bf M}_{\mathcal{H}}(\vec{x_i}) \oplus \Delta))].
\end{align}
We train HashTran-DNN with the final loss function:
\begin{equation}
Loss = \mathcal{L}_{\scriptsize{C}} + \lambda_{\scriptsize{D}}\mathcal{L}_{\scriptsize{D}},
\end{equation}
where $\lambda_{\scriptsize{D}} > 0$ is a hyper-parameter that is tuned to strength the DAE term via an exponential search~\cite{bergstra2012random}. This kind of training process makes the learned DNN classify slightly perturbed samples correctly and allows us to use the DAE to detect the out-of-distribution samples in the testing phase.

\subsubsection{Testing}
Since adversarial examples may be far from the support of the distribution of the training data, HashTran-DNN aims to reject such out-of-distribution samples (i.e., treating them as adversarial samples) before predicting the class (or a label) for the testing samples. 
The detection of out-of-distribution samples is based on the DAE reconstruction error, an idea inspired by Magnet~\cite{mengc_2017}. Because DAE encoding and reconstruction are operated on the training set, if a testing sample is drawn from the same distribution as the training samples, then a small reconstruction error is expected; otherwise, we can consider such testing sample as outliers.
Therefore, we need a threshold $t_{r}$ for flagging whether an input is out-of-distribution or not, which is a hyperparameter of the DAE. Intuitively, a smaller $t_{r}$ can help detect more adversarial examples, but runs into the risk of filtering out more normal samples, leading to a degradation in the classification accuracy. This suggests us to choose $t_{r}$ via a validation set of non-adversarial samples such that these samples can pass the filter at a high rate. 
HashTran-DNN allows to predict the class of testing samples if they can pass the DAE-based filter, as shown in the bottom of Figure \ref{fig:002}. 

\begin{remark}
Although HashTran-DNN focuses on enhancing the robustness of DNNs against adversarial samples, the framework can be equally applied to enhance the robustness of other machine learning models. 
Consider linear SVM under binary classification as an example. Given $\{{\bf H}_1,{\bf H}_2,\ldots,{\bf H}_L\}$ and an instance $\vec{x}$, the confidence $score$ can be defined as
\begin{equation}
 score =  {{w}^\textrm{T}}\cdot\lbrack{w}^\mathrm{T}_j\cdot{{\bf H}_j(\vec{x})} + b_j\rbrack + b\mathrm{;}\,~~(j=1\mathrm{,}\cdots\mathrm{,}L), \label{svm}
\end{equation}
where the ${w}$'s and $b$'s are weight vectors and  biases of SVMs. Eq.(\ref{svm}) says that $L$ internal SVMs learn and test over the corresponding hash transformation ${\bf H}_j(\cdot)$, and an external SVM aggregates their outputs to vote the final confidence with weight $w$. A majority of the $j$'s with ${\bf H}_j(\vec{x}')={\bf H}_j(\vec{x})$ lead to the robustness against adversarial example $\vec{x'}$. The $(L + 1)$ SVM models can be trained as usual~\cite{Hsieh:2008:DCD:1390156.1390208}.\end{remark}

\subsection{Analysis} \label{sec:analysis}

HashTran-DNN can accommodate any $\mathcal{H}$ that is {\em locality-preserving}, such as the 
aforementioned ${\bf H}_{\text{LSH}}$ constructed from {\it bit sampling} and the ${\bf H}_{\text{LNH}}$ obtained from Algorithm \ref{LFH_alg}. 
In the subsequent analysis, we make the following restrictions:
\begin{itemize}
\item Consider binary feature vectors, namely $\vec{x}\in\{0,1\}^n$; 
\item The distance function is the \textit{Hamming} distance, denoted by $d_H(\cdot,\cdot)$, namely $d_H(\vec{x_1}, \vec{x_2}) = {\|{\vec{x_1}-\vec{x_2}}\|}_0$ for 
$\vec{x_1},\vec{x_2} \in \{0,1\}^n$. The normalized Hamming distance is $\bar{d}_H(\vec{x_1}, \vec{x_2})=\frac 1n {\|{\vec{x_1}-\vec{x_2}}\|}_0$, namely the fraction of the coordinates where $\vec{x_1}$ and $\vec{x_2}$ are different. Note that $\bar{d}_H(\vec{x_1}, \vec{x_2})=1-{\rm Pr}(H_i(\vec{x_1})=H_i(\vec{x_2}))$.
\end{itemize}

In what follows we prove (i) the existence of a family of hashing transformations $\mathcal{H}$ with the locality-preserving property that satisfies Ineq. \eqref{geq:2Pr} and (ii) an approximation to distance-preserving property in an effort to approach Eq. \eqref{eq:keeping}.

\begin{theorem}[HashTran-DNN robustness]\label{Theo:robust} 
There exist hashing transformations such that Ineq. (\ref{geq:2Pr}) holds.
\end{theorem}

\begin{proof}
Consider a sample $\vec{x} \in \{0,1\}^n$ and its perturbed version $\vec{x}' \in \{0,1\}^n$ with ${\parallel\vec{x}'-\vec{x}\parallel}_0\leq\epsilon$. We have $0 \leq \bar{h}_H(\vec{x}', \vec{x}) \leq \frac \epsilon n$. When instantiating $\mathcal{H}$ as ${\bf H}_{\text{LSH}}$ or ${\bf H}_{\text{LNH}}$, $\mathcal{H}$ consists of $L$ independent hashing transformations $g^K_i$ for $1 \leq i \leq L$, implying 
\begin{align*}
&\mathbf{E}[\#\{j\in T_{\mathcal{H}}:{\bf H}_j(\vec{x})={\bf H}_j(\vec{x}')\}]\\\nonumber
&= L \times  {\mathrm{Pr}}\left(g^K_i(\vec{x})=g^K_i(\vec{x}')\right) = L \times P_1^K, \nonumber 
\end{align*}
 
where
\[
P_1 = 
\begin{cases}
1-\bar{d}_H(\vec{x}', \vec{x}) & ~\mbox{if $\mathcal{H}$ is instantiated as ${\bf H}_{\text{LSH}}$}\\

[1-\bar{d}_H(\vec{x'}, \vec{x})]^{m} & ~\mbox{if $\mathcal{H}$ is instantiated as ${\bf H}_{\text{LNH}}.$} 
\end{cases}
\]
By observation, we know that $K \times \mathrm{ln}(P_1) \geq \mathrm{ln}(\Theta) - \mathrm{ln}(L)$ is equivalent to $L \times P_1^K \geq \Theta$, which is equivalent to Ineq. (\ref{geq:2Pr}).
That is, there exist $\mathcal{H}$ such that 
Ineq. (\ref{geq:2Pr}) holds if and only if $K\leq{\frac{\mathrm{ln}(\Theta) - \mathrm{ln}(L)}{\mathrm{ln}(P_1)}}$.
\end{proof}

Theorem \ref{Theo:robust} says that for 
a desired threshold value $\Theta$, there exist hashing transformations with proper choices of parameters $K$ and $L$ under which Ineq. (\ref{geq:2Pr}) holds, implying classification robustness against adversarial samples generated by small perturbations.

As mentioned above, it is difficult to prove Eq.~(\ref{eq:keeping}) and therefore we consider the following weaker result: for any
$\vec{x_1},\vec{x_2} \in \{0,1\}^n$, a hashing transformation ${\bf H}_j$ does not change much of the normalized Hamming distance $\bar{d}_H(\vec{x_1},\vec{x_2})$ in the transformed space, namely that there exists some ${\bf H}_j$ and $\kappa>0$ such that
\begin{equation}
\label{weaker-accuracy-condition}
\mathbf{E}[|\bar{d}_H({\bf H}_j(\vec{x_1}),{\bf H}_j(\vec{x_2})) - \bar{d}_H(\vec{x_1},\vec{x_2})|]\leq \kappa. 
\end{equation}
Note that Eq. \eqref{weaker-accuracy-condition} is weaker than Eq. \eqref{eq:keeping}. 

\begin{theorem}[HashTran-DNN classification accuracy]\label{Theo:accuracy}
There exist hashing transformations such that Eq. \eqref{weaker-accuracy-condition} holds.
\end{theorem} 

\begin{proof}
In order to unify the presentation, let us uniformly denote the functions by $h_j$ when using LSH to instantiate $\mathbf{H}_j$, by $\text{DT}^{m,d}_{i,j}$ when using LNH to instantiate $\mathbf{H}_j$, by $h_j^m$ such that $h^1_j=h_j$ indicates LSH, and by $h_j^{m}$ with $m>1$ indicating LNH.
The normalized Hamming distance between $g_i^K(\vec{x_1})$ and $g_i^K(\vec{x_2})$ is
\[
\frac{1}{mK}\sum_{j=1}^{K}d_H(h_j^m(\vec{x_1}),h_j^m(\vec{x_2})). 
\]
Let $\bar{d}_H=\frac{1}{n}d_H(\vec{x_1},\vec{x_2})$ and $\bar{d}_{H,j}=\frac{1}{m}d_H(h_j^m(\vec{x_1}),h_j^m(\vec{x_2}))$. By doubling the value of parameter $K$ in the case of LSH $h_j$ or the case of LNH $\text{DT}^{m,d}_{i,j}$, we have 
\begin{align*}
&\quad\,\,\mathbf{E}\left[\left|\bar{d}_H - \frac{1}{2K}\sum_{j=1}^{2K}\bar{d}_{H,j}\right|\right] \\
&=\mathbf{E}\left[\left|\frac{1}{2}(\bar{d}_H-\frac{1}{K}\sum_{j=1}^{K}\bar{d}_{H,j}) 
+\frac{1}{2}(\bar{d}_H-\frac{1}{K}\sum_{j=K+1}^{2K}\bar{d}_{H,j})\right|\right]\\
&\leq\frac{1}{2}\mathbf{E}\left[\left|\bar{d}_H-\frac{1}{K}\sum_{j=1}^{K}\bar{d}_{H,j}\right|\right] 
+ \frac{1}{2}\mathbf{E}\left[\left|\bar{d}_H-\frac{1}{K}\sum_{j=K+1}^{2K}\bar{d}_{H,j}\right|\right] \\
&=\mathbf{E}\left[\left|\bar{d}_H-\frac{1}{K}\sum_{j=1}^{K}\bar{d}_{H,j}\right|\right].
\end{align*}
This leads to Eq. \eqref{weaker-accuracy-condition} and the theorem follows. 
\end{proof}

\section{Implementation and Evaluation} \label{sec:evl}

In this section we report our implementation of HashTran-DNN and evaluate its effectiveness via standard metrics (see, e.g., \cite{Pendleton:2016}) that include the classification accuracy (Acc), the False-Positive Rate (FPR), and the False-Negative Rate (FNR). 
 
\subsection{Dataset and feature extraction} 
We use an Android malware dataset that was collected from the Koodous Android malware analysis platform~\cite{Koodous:Online}. This dataset contains 49,829 Android malware samples and 48,406 benign Android samples. We treat the malware samples as non-adversarial samples. We split the dataset into three disjoint sets: a {\em training} set of 78,588 samples (including 39,914 malware samples and 38,674 benign samples), a {\em validation} set of 4,912 samples (including 2,475 malware samples and 2,437 benign samples), and a {\em testing} set of 14,735 samples (including 7,487 malware samples and 7,248 benign samples). 

In order to extract features of the samples, we use the Androguard \cite{Androguard:Online} to unpack Android Packages (APKs). We use the following kinds of static Android features: 
{\bf (i)} Permissions requested by an application (e.g., \textit{ android.permission.SEND\_SMS}). 
{\bf (ii)} Features indicating the application of hardware (e.g., \textit{android.hardware.wifi}).
{\bf (iii)} Names of application components including \textit{activity}, \textit{service}, \textit{broadcast receiver} and \textit{provider}.
{\bf (iv)} Intents {\it intent-filter} used to communicate with each other.
{\bf (v)} Permissions actually used for calling Application Programming Interface (API)~\cite{Daniel:NDSS}. These features are also considered by previous studies for Android malware detection \cite{Daniel:NDSS,Sayfullina7345283, grossePM0M16}.
The first four kinds of features can be extracted from the {\it AndroidManifest.XML} file, 
and the last kind of features
can be extracted by analyzing the disassembled code. The feature space is very large, containing 519,550 features in total. We propose reducing the feature space by removing the low-frequency features, which are the features that only occasionally appear in the samples. After removing the features with low-frequency ($<15$), we obtain 13,596 features in total.

\subsection{Waging the four attacks}
Recall that we use the $\ell_0$-norm to bound the degree of perturbation, namely
$$d_H(\vec{x},\vec{x}')={\Vert{\vec{x}'-\vec{x}}\Vert}_0\leq\epsilon,$$ 
where $\epsilon$ bounds the number of features that are perturbed. 
The {\em objective} of the attacker is to manipulate a malware sample $\vec{x}$ with $C(\vec{x})=1$ to an adversarial sample $\vec{x'}$ such that $C(\vec{x'})=0$.
In order to wage the four attacks (reviewed in Section \ref{adv_smps}), the attacker needs a surrogate DNN classifier helping generate adversarial examples. Since the attacker knows the training dataset and the feature set, 
the attacker can train its own surrogate model ~\cite{papernot_2016,mengc_2017}. In our experiments, we train (on behalf of the attacker) a surrogate DNN classifier with three hidden layers, where the first layer has 4,096 neurons, the second layer has 512 neurons, and the third layer has 32 neurons. In each of the four attack experiments, there are two steps: {\em selecting features to perturb}, {\em perturbing the selected features}, and {\em validating the perturbations}, which are elaborated below.

\subsubsection{Selecting features to perturb} 
In order to perturb malware samples while preserving malicious functionalities, we make the following observations. 
\begin{itemize}
\item We cannot delete or replace objects in the {\it AndroidManifest.xml} file because of the following two reasons. {\bf (i)} This file declares the objects (e.g., requested permission, application components and hardware) that must be declared to the Android operation system; otherwise, these objects and their associated functions will be ignored. {\bf (ii)} The declared application components are the public API that may be used by the other Android apps or the Android system, meaning that deleting or replacing them may crash the other apps or the Android system. 
\item Malware may require extra permissions, hardware resources, and unnecessary components from the Android system, suggesting us to insert {\it permission} requirement, {\it hardware} requirement, {\it activity}, {\it service} and {\it broadcast receiver} into {\it AndroidManifest.xml}. 
\end{itemize}
The preceding observations suggest us to insert some features into malware samples to generate adversarial samples.
Pertinent to the dataset, 11,863 (among the 13,596) features can be perturbed. In other words, any of these 11,863 features that is absent in a malware sample may be inserted by the four attack methods reviewed in Section~\ref{adv_smps}. 
Additional care needs to be taken for waging the CW and Mimicry attacks.

For waging the 
CW attack, which cannot be directly applied to the binary feature space---the context of the present paper. Therefore, we propose using the following variant of the CW attack. We use
\begin{align}
{\vec{x}'}&=\max(\textrm{clip}(\vec{x}'),\vec{x}) \circ \vec{v} + \vec{x} \circ (1 - \vec{v}) \nonumber\\ 
&= \max(\min(\max(\vec{x}',0),1), \vec{x}) \circ \vec{v} + \vec{x} \circ (1 - \vec{v}) \nonumber
\end{align}
to guide the perturbation when minimizing the loss function given by Eq. \eqref{eq:cw}, where $\max(\cdot,\cdot)$ and $\min(\cdot,\cdot)$ are respectively the element-wise maximum and minimum operation, and $\vec{v}$ ensures that the perturbation do not disrupt the malicious functionality (i.e., only the 11,863 features can be inserted). Once we obtain the optimization result, we use the nearest neighbor in the discrete space as $\vec{x'}$, while obeying the constraint that only the 11,863 features can be inserted.

For waging the Mimicry attack, we propose using the following heuristic to increase the effectiveness of the Mimicry attack.
Specifically, we use 60 benign samples to guide the perturbation of a single malware sample, leading to 60 adversarial samples; then, we select the adversarial sample (among the 60) that causes the lowest classification accuracy to the attacker's surrogate model (i.e., accommodating the worst-case scenario). 

Each of the four attack methods selects its own set of features to perturb or insert. The number of perturbed features in the JSMA attack, the GD-KDE attack, and the CW attack is bounded from above by the perturbation parameter $\epsilon$, which is the number of features that can be perturbed (pertinent to the context of binary feature representation). Nevertheless, the Mimicry attack implies that the number of perturbed features is not bounded because it mimics a benign sample as much as possible and parameter $\epsilon$ is not applicable in this case.

\subsubsection{Perturbing the selected features}
In the four attack experiments, we use the same set of 500 malware samples that are randomly select 
from the malware samples in the dataset; these 500 malware samples are all classified by the attacker's surrogate DNN model as malicious.
In each attack experiment, we insert into the malware samples the features that are respectively selected by the attack in question (as described in the previous step), and then we re-package the modified APK into adversarial samples.
In our experiment, we use the disassembly and repackage tool known as Apktool \cite{apktool:Online} and use the ElementTree API~\cite{ElementTree:Online} to modify the {\it AndroidManifest.xml} file for inserting those selected features. In each attack experiment, we successfully perturb  496 (among the 500) malware samples, while noting that the Apktool fails to disassemble the other 4 malware samples.
This means that we generate 496 adversarial samples in each of the four attacks.

\subsubsection{Validating the perturbations} 
In order to validate that the 496 adversarial samples are still malware, we use the dynamical malware analysis tool known as CuckooDroid sandbox~\cite{cuckoodroid:Online} to execute them.
For the sake of efficiency, we randomly select 23 adversarial examples from the 496 adversarial samples and confirm their maliciousness as follows. We use a sandbox to install the adversarial samples in an android emulator, monitor the adversarial sample's execution \cite{droidmon:Online}, and submit the adversarial sample to Virustotal~\cite{VirusTotal:Online}. These 23 adversarial samples are all deemed as malicious because each sample is detected by at least 10 detectors as malicious, which is perhaps acceptable according to a recent study on the trustworthiness of VirusTotal \cite{DBLP:journals/tifs/DuSCCX18}. 

\subsection{Evaluating the effectiveness of the HashTran-DNN defense}

The evaluation is centered at answering the following three Research Questions (RQ):
\begin{itemize}
\item RQ1: How effective is HashTran-DNN?
\item RQ2: How robust is HashTran-DNN when compared with other deep learning-based defense methods?
\item RQ3: What contributes to the effectiveness of HashTran-DNN? The answer to this question may be seen as a first step towards answering the much more difficult open problem: Why deep learning is effective? 
\end{itemize}

\subsubsection{RQ1: How effective is HashTran-DNN?}

For answering RQ1, we set $\epsilon=10$ for the JSMA, GD-KDE, and CW attacks, meaning that at most 10 features are perturbed.
Since HashTran-DNN uses hashing transformations and possibly a DAE, we consider four combinations in terms of the  $\mathcal{H}$ instantiation (as LSH or LNH) and whether or not to use DAE, namely: {\bf (i)} HashTran-DNN with LSH, where LSH is derived from the bit sampling method; {\bf (ii)} HashTran-DNN with LNH, where LNH is derived from Algorithm \ref{LFH_alg}; {\bf (iii)} HashTran-DNN with LSH-DAE; {\bf (iv)} HashTran-DNN with LNH-DAE. 

For hyper-parameters, we focus on tuning the following:
\begin{itemize}
\item $K$: the number of hash function $h_j$ or ${\textrm DT}^{m,d}_{i,j}$ respectively in $g^K_{\rm LSH}$ or $g^K_{\rm LNH}$, where $1 \leq j \leq K$.
\item $L$: the number of $g^K_{\rm LSH}$ or $g^K_{\rm LNH}$ respectively in $\mathbf{H}_{\text{LSH}}$ or $\mathbf{H}_{\text{LNH}}$.
\end{itemize}
For training the four kinds of HashTran-DNN models, 
we use the Adam optimization method with batch size 128, epochs 30, dropout rate 0.4, and learning rate 0.001. We select the model that achieves the highest classification accuracy on the afore-mentioned validation set of 4,912 samples. 
The other hyper-parameters are selected as follows.

\begin{table}[!htbp]
\caption{Acc (accuracy) of HashTran-DNN with LSH when applied to the {testing set of 14,735 samples (containing no adversarial samples, the 2nd column) and when applied to the 496 adversarial samples respectively generated by the four attacks (the 3rd to 6th columns), under different choices of $m$.}}
		\centering
		\begin{tabular}{ccccccc}
			\hline
			 {\multirow{2}{*}{$m$}}&{\multirow{2}{*}{Acc (\%)}}&\multicolumn{4}{c}{Acc (\%) under attack}\\\cline{3-6}
			 &&JSMA&GD-KDE&CW&Mimicry\\ 
            \hline
             15&91.30&68.64&81.72&84.17&68.05 \\
            \hline
             30&91.40&71.79&92.52&91.32&65.35 \\
            \hline
             60&91.54&90.43&92.89&91.88&69.56 \\
            \hline
             100&91.75&93.97&93.50&93.92&74.41 \\
            \hline
             116&91.71&93.55&93.72&95.14&74.05 \\
		    \hline
		\end{tabular}
		\label{tab:analysis_m_prv}
\end{table}

For training a model of HashTran-DNN with LNH, we set the height of the Decision Trees as $d=4$.
For selecting $m$, 
we set $(K,L)=(32,64)$ and vary $m$ from $2^d - 1$ to $\sqrt{n}$, where $n$ is the dimension of the feature space. Table~\ref{tab:analysis_m_prv} summaries the results, and shows that $m=15$ leads to the lowest classification accuracy of the HashTran-DNN model against the JSMA, GD-KDE, and CW attacks. 
As $m$ increases, the classification accuracy with respect to the original testing set (containing no adversarial samples) steadily increases, albeit slightly (the 2nd column). However, the classification accuracy against the 496 adversarial samples, which are respectively generated by the four attacks (the 3rd to 6th columns), varies substantially.
These observations suggest us to set  $m\approx\sqrt{n}$.

\begin{figure}[!htbp]
  \begin{subfigure}[b]{0.245\textwidth}
    \includegraphics[width=\textwidth]{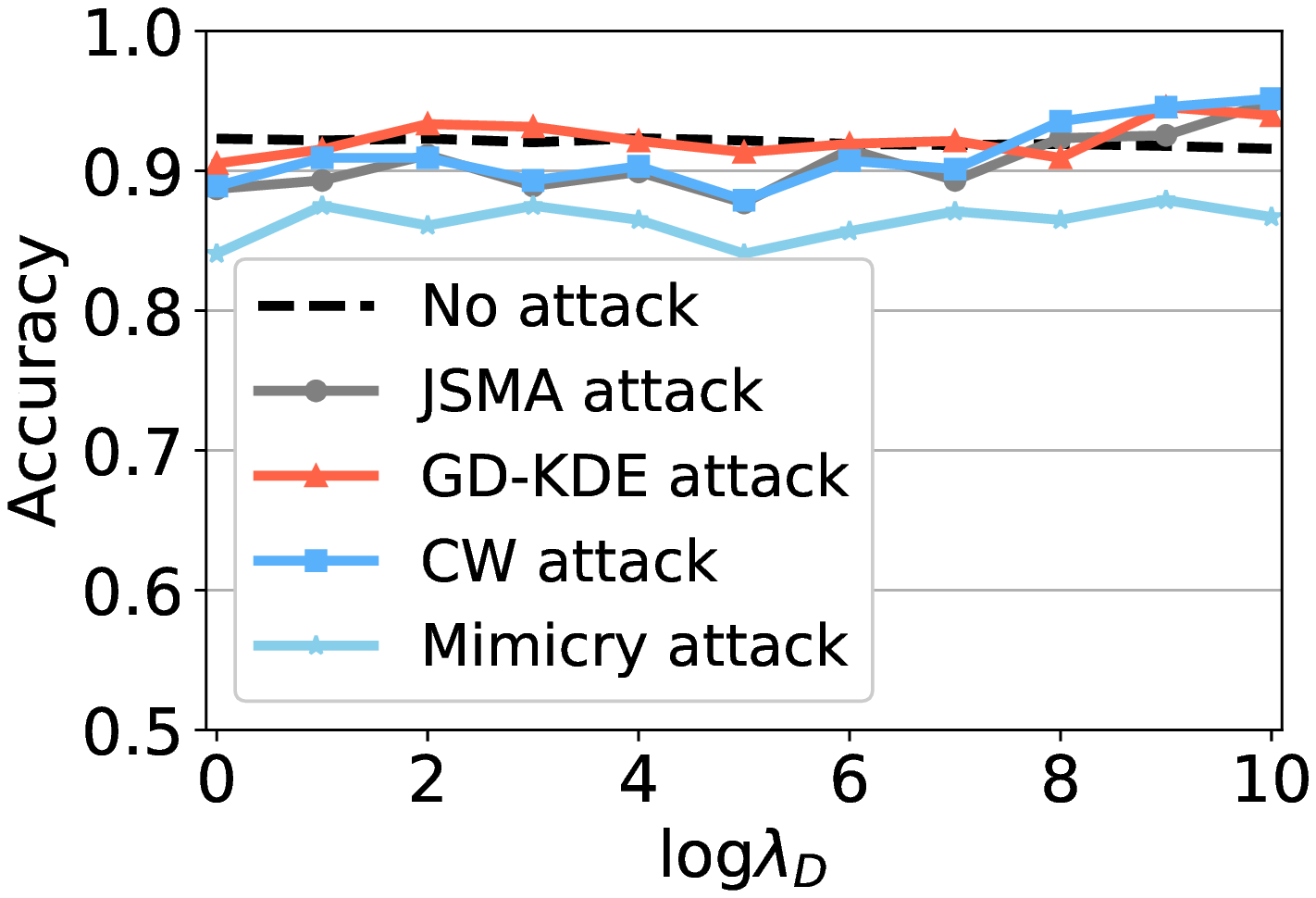}
    \caption{{\footnotesize HashTran-DNN w/ LSH-DAE}}
    \label{fig:lsh-dae-lam}
  \end{subfigure}
  \begin{subfigure}[b]{0.23\textwidth}
    \includegraphics[width=\textwidth]{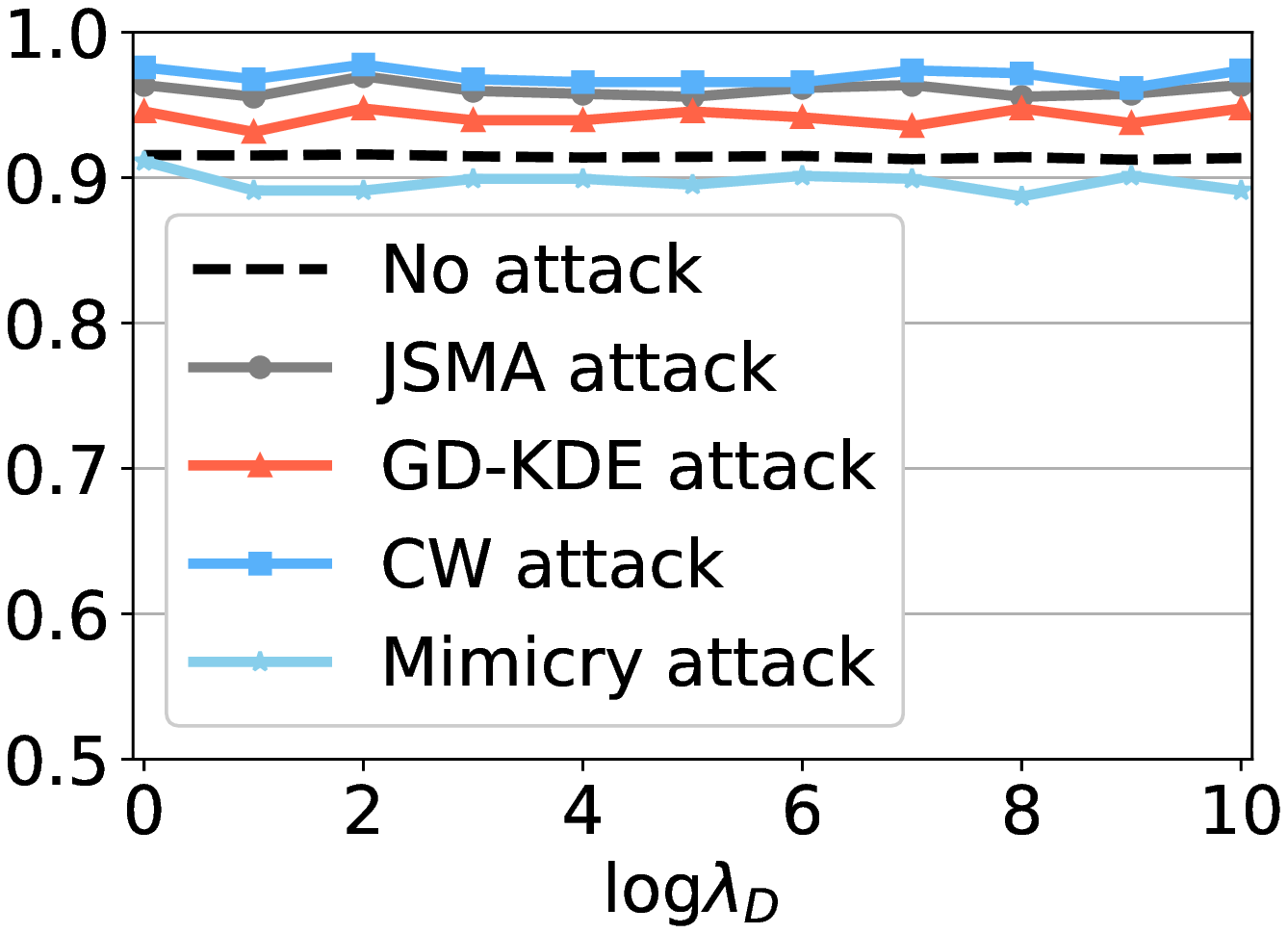}
    \caption{{\footnotesize HashTran-DNN w/ LNH-DAE}}
    \label{fig:lnh-dae-lam}
  \end{subfigure}
\caption{Accuracies of HashTran-DNN with LSH-DAE and HashTran-DNN with LNH-DAE when applied to the testing set of 14,735 samples (containing no adversarial samples) and when 
applied to the 496 adversarial samples that are respectively generated by the four attacks, under different choice of $\lambda_D$.}
  \label{fig:lambda-d}
\end{figure}

For training HashTran-DNN with LSH-DAE and HashTran-DNN with LNH-DAE, we sample the probability for noise injection from $\max(0,\mathcal{N}(0,{(\left.{10}\middle/{n} \right.)}^2))$, where $\mathcal{N}(\cdot,\cdot)$ is the Gaussian Distribution. We set the threshold $t_r$ to make 99.9\% of the validation set of 4,912 samples 
pass the DAE-based detector. We select $\lambda_D$ as follows.
\begin{itemize}
\item In the case of HashTran-DNN with LSH-DAE, we set $(K,L)=(128,128)$ and vary $\lambda_D$ from 1 to 1024 exponentially.
From Figure \ref{fig:lsh-dae-lam}, we observe that the accuracy against the JSMA, GD-KDE and CW attacks increases with $\lambda_D$, while noting that the accuracy on the testing set {of 14,735 samples (containing no adversarial samples)} drops slightly. This suggests us to select $\lambda_D=256$ for HashTran-DNN with LSH-DAE. 
\item In the case of HashTran-DNN with LSH-DAE, we set $(K,L)=(32,32)$ and vary $\lambda_D$ from 1 to 1024 exponentially. Figure \ref{fig:lnh-dae-lam} shows that the accuracy under each attack varies slightly with $\lambda_D$. Therefore, we select $\lambda_D=1$.
\end{itemize}

\begin{table*}[!htbp]
\caption{HashTran-DNN Acc (accuracy) without using adversarial samples (testing set of 14,735 samples, the 3rd column) vs. using 496 adversarial malware samples respectively (the 6th to 9th columns) under different hyper-parameters $(K,L)$.}
		\centering
		\begin{tabular}{c|c|ccc|ccc|c}
			\hline
\multirow{2}{*}{Classifiers}&{\multirow{2}{*}{$(K,L)$}}&{Acc (\%)}&{FNR (\%)}&{FPR (\%)}&\multicolumn{3}{c|}{Acc (\%) with $\epsilon$=10} & {Acc (\%)} \\\cline{3-9}
&& \multicolumn{3}{c|}{absence of adversarial samples}&JSMA attack&GD-KDE attack&CW attack&Mimicry attack \\
             \hline
	\multirow{7}{*}{HashTran DNN w/ LSH} &(32,64)&88.25&18.82&4.42&75.87&70.26&67.16&50.47\\
    					 &(64,64)&91.15&13.87&3.67&81.56&76.62&70.09&49.89\\
                         &(32,128)&91.13&14.24&3.35&86.00&81.11&70.56&52.26\\
    					 &(64,128)&91.79&13.89&2.35&71.93&82.17&77.48&51.48\\
    					 &(128,128)&92.19&12.46&2.83&64.51&73.02&68.83&52.85\\
                         &{(128,256)}&92.11&13.22&2.31&64.12&60.69&71.23&52.93\\
                         &(256,256)&92.61&12.87&1.74&65.99&59.35&65.52&47.71\\
               \hline
	\multirow{3}{*}{HashTran DNN w/ LNH} &(32,32)&91.69&14.01&2.46&90.76&94.31&93.16&73.74\\
                         &{(32,64)}&91.71&14.07&2.35&93.55&93.72&95.15&74.05\\
                         &(64,64)&91.54&14.24&2.45&95.93&93.92&95.76&79.66\\
				\hline
	\multirow{3}{*}{HashTran DNN w/ LSH-DAE} &(128,128)&91.92&13.41&2.57&92.34&90.93&93.55&86.49 \\
							 &{(128,256)}&92.28&12.54&2.75&94.35&89.92&92.14&88.31\\
							 &(256,256)&92.50&12.46&2.37&91.73&89.52&87.50&87.10 \\
				\hline
	\multirow{3}{*}{HashTran DNN w/ LNH-DAE} &(32,32)&91.56&13.73&2.98&96.37&94.56&97.58&91.13 \\
							 &(32,64)&{91.73}&13.62&2.73&96.98&95.16&97.18&92.94 \\
							 &{(64,64)}&91.65&15.12&1.37&96.17&94.15&96.57&87.70 \\
				\hline
		\end{tabular}
		\label{tab:analysis}
\end{table*}

Table~\ref{tab:analysis} summarizes the evaluation result.
We make the following observations.
First, for HashTran-DNN with LSH, the classification accuracy increases with $K$ (the number of hash function in $g^K$),
which confirms Theorem \ref{Theo:accuracy} in Section \ref{sec:analysis}, namely that increasing $K$ can improve the classification accuracy against non-adversarial malware samples. When a higher accuracy ($\geq 91.79\%$) is achieved {against the testing set containing no adversarial samples}, the accuracy drops substantially against any of the JSMA, GD-KDE, and CW attacks.
Second, HashTran-DNN with LNH, although achieving a slightly lower classification accuracy in some cases than HashTran-DNN with LSH against the testing set that contains no adversarial examples, the former achieves a much higher classification accuracy ($\geq 90.76\%$) than the latter against the JSMA, GD-KDE, and CW attacks as well as a 20\% increase in the classification accuracy against the Mimicry attack. This can be attributed to the fact that the former treats features more equally than the latter. 
Third, using DAE can further improve the classification accuracy against the four attacks. This is especially true for 
HashTran-DNN with LNH-DAE, which increases, for example,
the classification accuracy against the CW attack from 
68.83\% to 93.55\%.

\begin{insight}
\label{insight-RQ1}
HashTran-DNN with LNH can effectively defend against the JSMA, GD-KDE, and CW attacks.
Moreover, 
HashTran-DNN with LNH-DAE can effectively defend against all of the four attacks.
\end{insight}

\subsubsection{RQ2: How robust is HashTran-DNN when compared with other defense methods?}

We compare HashTran-DNN with the RFN~\cite{wang_2017} and iterative Adversarial Training defense methods~\cite{DBLP:journals/corr/KurakinGB16a}  reviewed in Section \ref{sec:review-defense-methods}.
We conduct five experiments: (i) Standard DNN; (ii) the RFN defense; (iii) the Adversarial Training defense; (iv) HashTran-DNN with LSH-DAE; and (v) HashTran-DNN with LNH-DAE (noting that the last two are chosen because they are respectively more effective than HashTran-DNN with LSH and HashTran-DNN with LNH). In each experiment, we consider five scenarios: the testing set of 14,735 samples (containing no adversarial samples) as well as the 496 adversarial samples that are respectively generated by the JSMA, GD-KDE, CW, and Mimicry attacks. In order to see the impact of the degree $\epsilon$ of perturbation, we consider $\epsilon=10,20,30$ for the JSMA, GD-KDE, and CW attacks (while recalling that $\epsilon$ is not applicable to the mimicry attack).

\begin{table*}[!htbp] 
	\caption{Hyper-parameters used in the experiments for answering RQ2. 
    }
	\centering
	\begin{tabular}{cccccccc}
	  \hline
		\multirow{2}{*}{Defense} & \multicolumn{7}{c}{Hyper-parameters} \\\cline{2-8}
		                                       &{DNN architecture}&{Activation}&{Optimizer}&{Learning rate}&{Dropout rate}&{Batch size}&{Epoch}\\
		\hline
		{No defense (standard DNN)}&{4096-512-32-2}&{Relu}&{Adam}&{0.001}&{0.4}&{128}&{30}\\
		\hline\hline
		{RFN}&{4096-512-32-2}&{Relu}&{Adam}&{0.001}&{0.4}&{128}&{30}\\
		\hline
		{Adversarial Training}&{4096-512-32-2}&{Relu}&{Adam}&{0.001}&{0.4}&{128}&{30}\\
		\hline\hline
		{HashTran-DNN w/ LSH-DAE}&{256,128-512-32-2}&{Relu}&{Adam}&{0.001}&{0.4}&{128}&{30}\\
		\hline
		{HashTran-DNN w/ LNH-DAE}&{64,128-512-32-2}&{Relu}&{Adam}&{0.001}&{0.4}&{128}&{30}\\
		\hline
	\end{tabular}
	\label{tab:hp}
\end{table*}

Table~\ref{tab:hp} summarizes the resulting neural network structures and hyper-parameters, while the other parameters are described as follows.
We select the nullification rate in the RFN defense by sampling from the Gaussian Distribution $\mathcal{N}(0.3,0.05^2)$, while making the adversarial training penalize the adversarial spaces {searched by} the JSMA attack method ($\epsilon=$10) iteratively. According to Table~\ref{tab:analysis}, we set $(K,L)=(128,256)$ and $(K,L)=(32,64)$ for HashTran-DNN with LSH-DAE and HashTran-DNN with LNH-DAE, respectively. We select the model that 
achieves the highest classification accuracy against the validation set of 4,912 samples (containing no adversarial samples).

\begin{table*}[!htbp] 
	\caption{Classification accuracy against the testing set of 14,735 samples containing no adversarial samples (No attack) and the 496 adversarial samples respectively generated by the JSMA, GD-KDE, CW, and Mimicry attacks.} 
	\centering
	\begin{tabular}{c|c|ccc|ccc|ccc|c}
	  \hline
		\multirow{2}{*}{Defense}&\multirow{2}{*}{No attack}&\multicolumn{3}{c|}{JSMA attack}&\multicolumn{3}{c|}{GD-KDE attack}&\multicolumn{3}{c|}{CW attack}&\multirow{2}{*}{Mimicry} \\\cline{3-11}
&&{$\epsilon=10$}&{$\epsilon=20$}&{$\epsilon=30$}&{$\epsilon=10$}&{$\epsilon=20$}&{$\epsilon=30$}&{$\epsilon=10$}&{$\epsilon=20$}&{$\epsilon=30$}& \\
		\hline
{No defense (standard DNN)}&{92.50}&{58.40}&{41.71}&{35.13}&{51.11}&10.88&0.614&57.34&38.70&5.552&13.00\\
		\hline\hline
		 {RFN}&91.93&{68.72}&{45.81}&{40.07}&{68.42}&37.67&12.17&70.87&53.63&35.71&26.14\\	
		\hline
		 {Adversarial Training}&92.18&{98.94}&{99.78}&{100.0}&{98.94}&{99.58}&100.0&98.74&99.78&100.0&85.41\\														
		\hline\hline
		 {HashTran-DNN w/ LSH-DAE}&92.28&{94.35}&{100.0}&{100.0}&{89.92}&94.56&96.17&92.14&98.59&99.80&88.31\\
			\hline
		 {HashTran-DNN w/ LNH-DAE}&91.73&96.98&96.77&96.77&{95.16}&{94.56}&{93.55}&97.18&95.56&93.35&92.94\\
		\hline
	\end{tabular}
	\label{tab:cmp}
\end{table*}

Table \ref{tab:cmp} summarizes the experimental result. We make the following observations. First, the standard DNN is vulnerable to the four attacks, and the higher the perturbation (when applicable), the lower the classification accuracy. Second, the four defense methods incur no significant side-effects in the absence of adversarial examples, meaning that they can be used even if the attacker does not launch adversarial samples. This matter is important in dealing with the uncertainty that in the real world, the defender does not know for certain when the attacker will launch adversarial samples. Third, the RFN defense cannot effectively defend against any of the four attacks. The adversarial training defense is highly effective against the JSMA, GD-KDE, and CW attacks, but not very effective against the Mimicry attack. 
In contrast, HashTran-DNN with LSH-DAE and HashTran-DNN with LNH-DAE achieve a classification accuracy that is comparable to what is achieved by the adversarial training defense against the JSMA, GD-KDE, and CW attacks (above 92.14\%, except for the GD-KDE attack with $\epsilon=10$, which leads to a 89.92\% classification accuracy). Moreover, both HashTran-DNN defense methods can more effectively defend against the Mimicry attack than the adversarial training defense, with a 7.53\% increase in the case of HashTran-DNN with LNH-DAE. 
{We reiterate that this effectiveness is achieved without using adversarial samples to train the HashTran-DNN models.}

\begin{insight}
Standard DNNs can be ruined by adversarial malware samples. RFN is not effective against any of the four attacks. 
Adversarial Training is effective against the JSMA, GD-KDE, and CW attacks, but not effective against the Mimicry attack. HashTran-DNN, while not using adversarial samples in training, is effective against the four attacks.
\end{insight}

\subsubsection{RQ3: What contributes to the effectiveness of HashTran-DNN?}

\begin{figure*}[!htbp]
\centering
  \begin{subfigure}[b]{0.25\textwidth}
    \includegraphics[width=\textwidth]{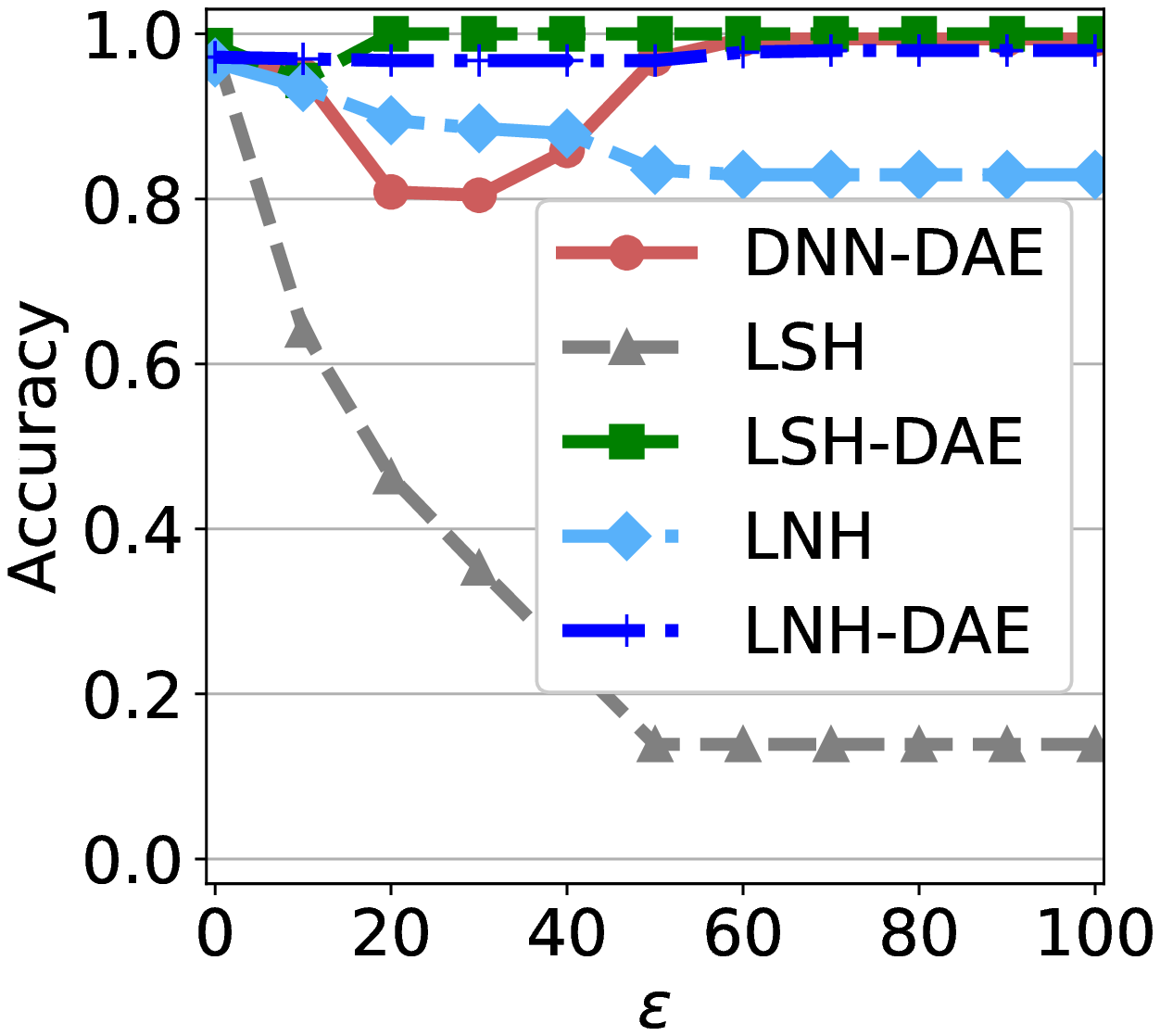}
    \caption{JSMA attack}
    \label{fig:jsma}
  \end{subfigure}
  \begin{subfigure}[b]{0.235\textwidth}
    \includegraphics[width=\textwidth]{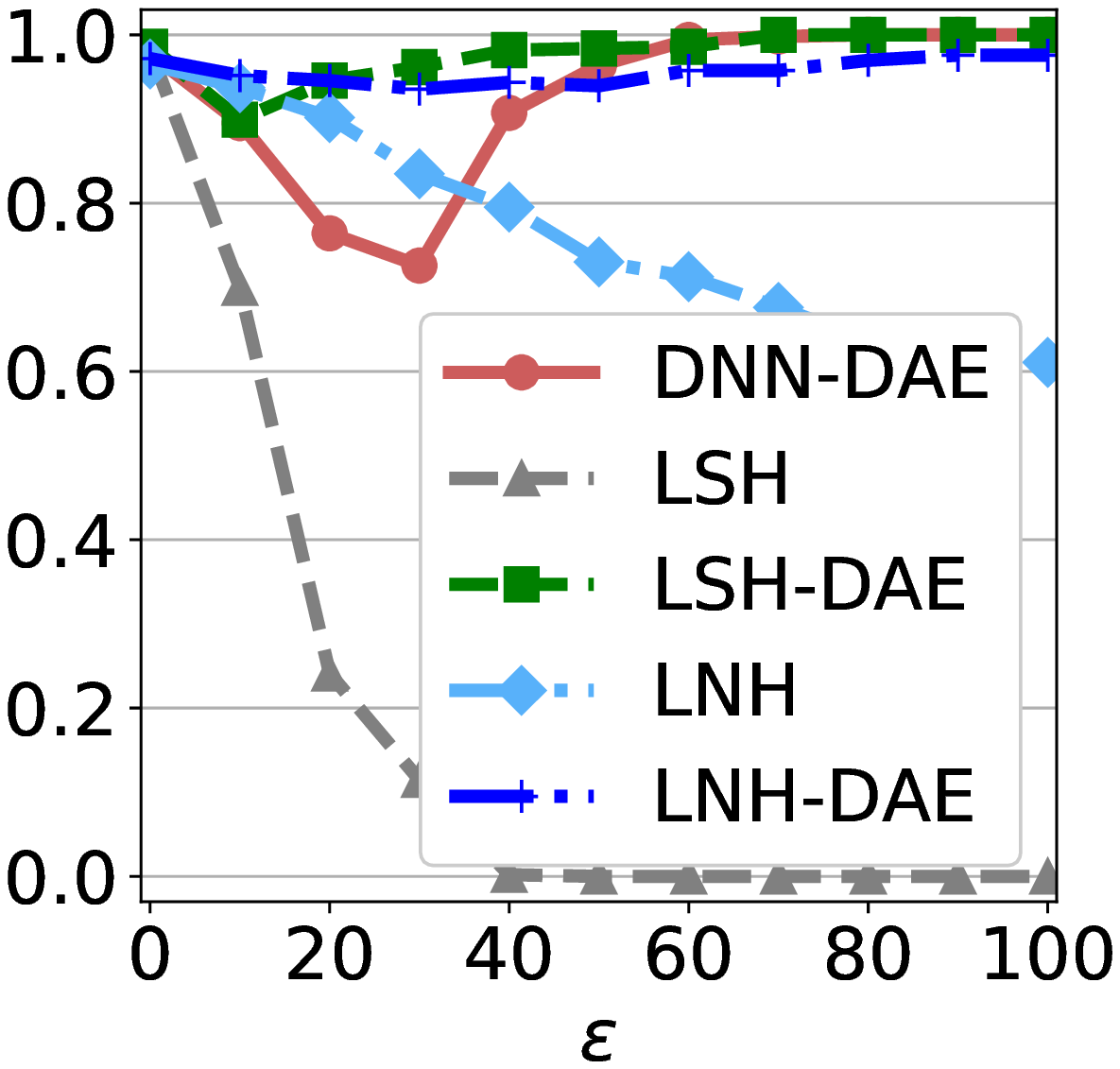}
    \caption{GD-KDE attack}
    \label{fig:gdkde}
  \end{subfigure}
  \begin{subfigure}[b]{0.235\textwidth}
    \includegraphics[width=\textwidth]{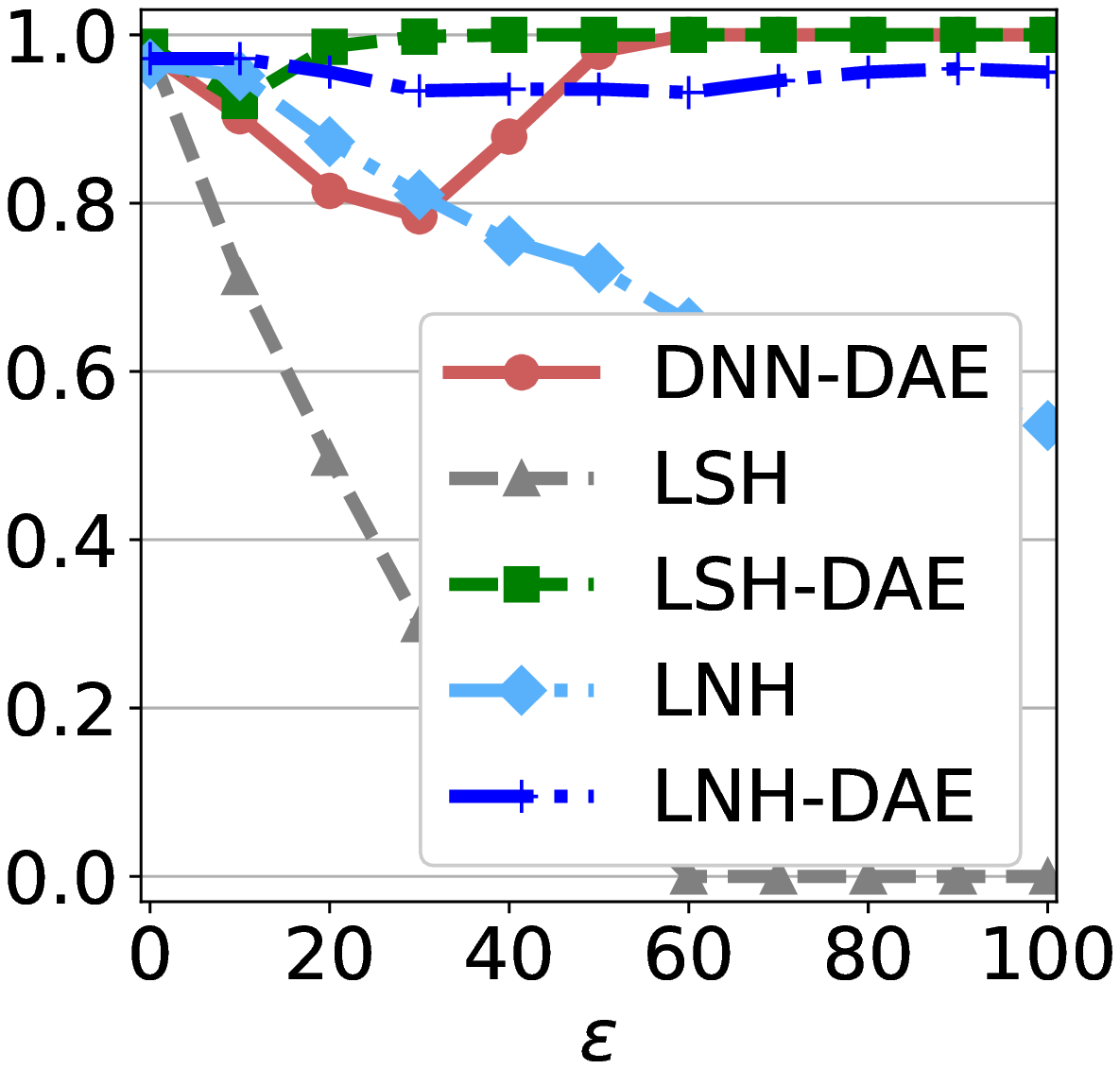}
    \caption{CW attack}
    \label{fig:cw}
  \end{subfigure}
  \begin{subfigure}[b]{0.225\textwidth}
    \includegraphics[width=\textwidth]{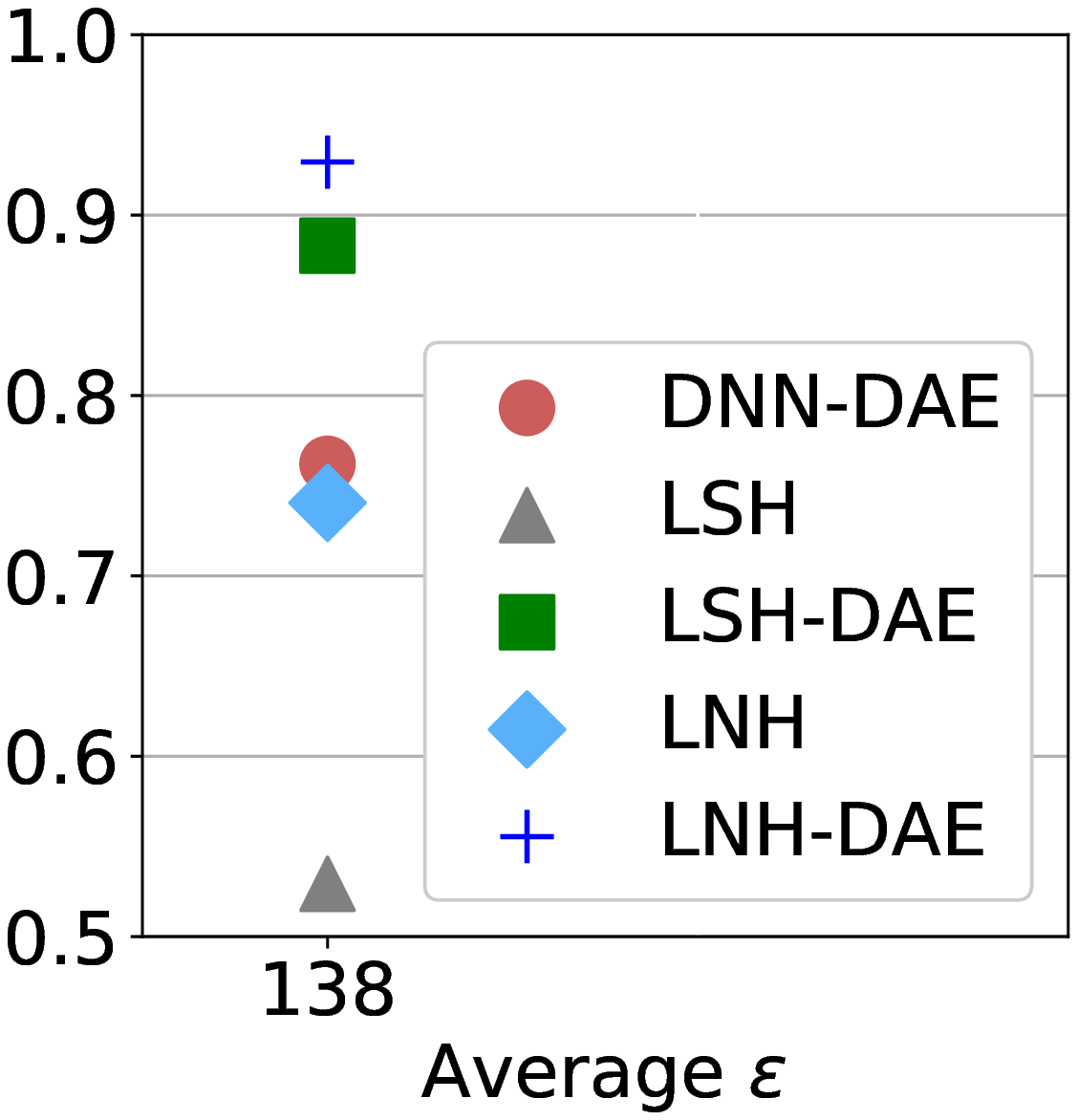}
    \caption{Mimicry attack}
    \label{fig:mimicry}
  \end{subfigure}
	\caption{Classification accuracy against the 496 adversarial samples that are respectively generated by the JSMA, GD-KDE, CW, and Mimicry attacks with different degrees of perturbations, while noting that for the Mimicry attack, the degree of perturbation is not an input parameter but averaged over the actual perturbations (because $\epsilon$ is not applicable). 
    }
  \label{fig:hashTran-work}
\end{figure*}

It is an open problem to explain the effectiveness of deep learning models. Nevertheless, we can at least get some insights into the effectiveness of HashTran-DNN. For this purpose, we conduct four experiments corresponding to the four attacks. In each experiment, we consider five DNN models: (i) Standard DNN with DAE, denoted by DNN-DAE; (ii) HashTran-DNN with LSH; (iii) HashTran-DNN with LNH; (iv) HashTran-DNN with LSH-DAE; and (v) HashTran-DNN with LNH-DAE. The idea is that by comparing the classification accuracy of (i), which does not using the hash representation, and that of (ii)-(iii), which does use the hash representation, we can observe the contribution of the hash representations to the classification accuracy. 

In the DNN-DAE experiment, the hyper-parameters are the same as the HashTran-DNN with LSH-DAE experiment. For fair comparison, we train a DNN-DAE model, which achieves a 92.09\% accuracy on the testing set of 14,735 samples (contains no adversarial samples), while noting that this accuracy is comparable to that of the standard DNN model (92.50\%). In the experiments of HashTran-DNN with LSH and HashTran-DNN with LSH-DAE, we set $(K,L)=(128,256)$ because as shown in Table~\ref{tab:analysis}, this combination leads to a good classification accuracy on the original testing set (containing no adversarial samples). In the experiments of HashTran-DNN with LNH and HashTran-DNN with LNH-DAE, we set $(K,L)=(32,64)$ because as shown in Table~\ref{tab:analysis}, this combination leads to the highest classification accuracy on the original testing set (containing no adversarial samples).

Figure \ref{fig:hashTran-work} plots the classification accuracy of the five experiments mentioned above.
We make the following observations.
First, the contribution of LSH and LNH hashing to the classification accuracy decreases as the degree of perturbation increases in the JSMA, GD-KDE, and CW attacks.
Second, the use of DAE can offset the incapability of hashing in coping with a high degree of perturbation. This can be attributed to the fact that the DAE can filter testing samples that are far away from the distribution of the training samples.
Third, there is a substantial drop in the classification accuracy of the DNN-DAE model against the 496 adversarial samples that are respectively generated by the JSMA, GD-KDE, and CW attacks with perturbation bound at $\epsilon =30$. However, the classification accuracy increase when $\epsilon$ increases above $\epsilon =30$. This phenomenon is not exhibited by the HashTran-DNN with LSH-DAE and HashTran-DNN with LNH-DAE. We attribute this discrepancy to the following: On one hand, the effectiveness of DAE against {\em small}, but not {\em large}, perturbations (i.e., the adversarial samples are close to the distribution of the training set) is known for its instability \cite{gur_2014}, 
explaining the phenomenon exhibited by the DNN-DAE model. On the other hand, the instability of DAE is eliminated by the HashTran-DNN with LSH-DAE and HashTran-DNN with LNH-DAE because the hashing transformation (or the hash representation) regularizes the corresponding HashTran-DNN to capture the locality information in the latent space.

\begin{insight}
The effectiveness of HashTran-DNN in detecting adversarial malware examples comes from two aspects: the hashing transformation, which helps DNNs cope with small perturbations, and the DAE, which regularizes DNN and filters large perturbations.
\end{insight}

\section{Limitations} \label{sec:lmt}
\label{sec:limitations}

The present study has several limitations.
{\bf (i)} The framework considers DNN only, meaning that it needs to be extended to accommodate other kinds of deep learning models. This is by no means straightforward.
{\bf (ii)} Our instantiations of HashTran-DNN consider two hashing transformations, namely LSH and LNH, while recalling that the specification of LNH is constructive rather than an explicit expression. Future research needs to define LNH explicitly and consider possibly other hash functions that can lead to better results.
{\bf (iii)} The present study focuses on {\em gray-box} attacks. Future research needs to investigate more rigorous {\em white-box} attacks, in which all parameters are exposed to the attacker.
{\bf (iv)} The present study focuses on malware classification, which is our interest and original motivation. It is interesting to investigate how the HashTran-DNN framework may be applied to image processing and other application domains.

\section{Conclusion} \label{sec:conclusion}

We have presented the HashTran-DNN framework for making DNN-based malware classifiers robust against adversarial malware samples. The framework is centered at using locality-preserving hash transformations and DAE to reduce, if not eliminate, the effect of adversarial perturbations. Experimental results show that the framework can effectively defend against the four attacks.

Future research problems are abundant. In addition to the limitations mentioned in Section \ref{sec:limitations}, it is an outstanding open problem to fully characterize the implications of the locality-preserving property in defending against adversarial samples.

\ifCLASSOPTIONcaptionsoff
  \newpage
\fi



\bibliographystyle{IEEEtran}
\bibliography{malware_paper}
\end{document}